\documentclass[conference,compsoc]{IEEEtran}

\usepackage[T1]{fontenc}
\usepackage[utf8]{inputenc}
\usepackage{amsmath,amssymb,mathtools}
\usepackage{amsthm}
\usepackage{array}
\usepackage{booktabs}
\usepackage{tabularx}
\usepackage{enumitem}
\usepackage{graphicx}
\usepackage{microtype}
\usepackage{multirow}
\usepackage{pifont}
\usepackage{tikz}
\usetikzlibrary{arrows.meta,positioning,fit,calc,backgrounds}
\usepackage{xcolor}
\usepackage[most]{tcolorbox}
\usepackage{xspace}
\usepackage[hidelinks]{hyperref}
\hypersetup{
  pdftitle={When Does Authorization End? Effect Closure at Provider Boundaries},
  pdfauthor={Igor Santos-Grueiro}
}

\newcommand{\system}{\mbox{\textsc{EffectBound}}\xspace}
\newcommand{\cmark}{\ding{51}}
\newcommand{\xmark}{\ding{55}}
\newcommand{\realizable}{\mathsf{Realizable}}
\newcommand{\requiredcomplete}{\mathsf{RequiredComplete}}
\newcommand{\contractsettled}{\mathsf{ContractSettled}}

\providecommand{\Description}[1]{}

\newtheorem{definition}{Definition}
\newtheorem{proposition}{Proposition}
\newtheorem{theorem}{Theorem}
\newtheorem{corollary}{Corollary}

\newtcolorbox{rqanswer}[1]{
  enhanced jigsaw,
  lower separated=false,
  colback=white!90!gray,
  colframe=white!20!black,
  colbacktitle=white!50!gray,
  coltitle=black,
  arc=0.8mm,
  outer arc=0.8mm,
  boxrule=1.1pt,
  title={#1 summary},
  fonttitle=\bfseries\small,
  fontupper=\small,
  attach boxed title to top left={xshift=0.5cm,yshift=-2mm},
  boxed title style={
    colback=white!50!gray,
    colframe=white!20!black,
    boxrule=1.1pt,
    arc=0.8mm,
    outer arc=0.8mm,
    left=3pt,
    right=3pt,
    top=1pt,
    bottom=1pt
  },
  left=4pt,
  right=4pt,
  top=6pt,
  bottom=2.5pt,
  before skip=8pt plus 1pt minus 1pt,
  after skip=5pt plus 1pt minus 1pt,
  pad at break*=2pt
}

\newtcolorbox{effectcontract}{
  enhanced jigsaw,
  lower separated=false,
  colback=white,
  colframe=black!78,
  colbacktitle=black!78,
  coltitle=white,
  arc=0.8mm,
  outer arc=0.8mm,
  boxrule=1.1pt,
  title={Effect-closure contract},
  fonttitle=\bfseries\small,
  fontupper=\small,
  attach boxed title to top left={xshift=0.15in,yshift=-0.1in},
  boxed title style={
    colback=black!78,
    colframe=white,
    boxrule=0pt,
    arc=0.6mm,
    outer arc=0.6mm,
    left=3pt,
    right=3pt,
    top=1pt,
    bottom=1pt
  },
  left=4pt,
  right=4pt,
  top=6pt,
  bottom=3pt,
  before skip=8pt plus 1pt minus 1pt,
  after skip=5pt plus 1pt minus 1pt
}

\newtcolorbox{certificatebox}[1]{
  enhanced jigsaw,
  colback=white,
  colframe=black!70,
  colbacktitle=black!10,
  coltitle=black,
  arc=0.7mm,
  outer arc=0.7mm,
  boxrule=1.1pt,
  title={#1},
  fonttitle=\bfseries\small,
  fontupper=\small,
  left=4pt,
  right=4pt,
  top=3pt,
  bottom=2pt,
  before skip=4pt plus 1pt minus 1pt,
  after skip=5pt plus 1pt minus 1pt,
  pad at break*=0pt
}

\title{When Does Authorization End?
 Effect Closure at Provider Boundaries}
\author{
\IEEEauthorblockN{Igor Santos-Grueiro}
\IEEEauthorblockA{International University of La Rioja\\
igor.santosgrueiro@unir.net}
}

\begin{document}

\maketitle

\begin{abstract}
Revocation completion, clean state, or operation success can leave authorized work able to cause an effect the application rejects while the provider stays within its contract. We call the absence of all such paths \emph{policy-relative effect closure}, or \emph{effect closure} for short. Thus, a grant is closed when its existing authorizations retain no such path, and it cannot issue any new ones.

 We present \system, which uses an evidence-supported finite contract to decide whether an interface can truthfully report closure while required work completes. It reduces this to finite control with hidden state and returns a strategy, an impossibility certificate, or no verdict when evidence is
insufficient. Machine-checked proofs establish the reduction and checker soundness; the checker derives closure results and validates certificates.

Across GitHub, Kubernetes, NATS, and Kafka, closure fails in three ways: an interface lacks a needed control, clean visible state hides active work, or the model stops before the \emph{effect frontier}---the last point where the effect can be prevented. The GitHub tool cannot bind a merge to the reviewed
commit; a controlled run confirms that it may merge a different commit. NATS can report no stored or pending messages while dispatched work can still publish downstream. In Kafka, all fixed-set brokers had applied the revocation, yet an earlier authorized request could still append. We add a gate that
blocks new use of revoked authority and delays return until earlier in-flight work completes. In a fixed-set Kafka~4.3.1 test deployment, this closes the studied synchronous, nontransactional write path without blocking unrelated requests. For a grant, authorization ends only when issuance stops and no
earlier authorization can reach an effect the application rejects.

\end{abstract}

\section{Introduction}

\emph{When can a provider boundary truthfully report that previously authorized
work can no longer cause an effect the application rejects?}
We trace six paths from authorization to application effect in GitHub,
Kubernetes, NATS, and Kafka. Consider an application using GitHub's Model
Context Protocol (MCP) server to merge pull request~42. The application approves
$h_1$ and calls \texttt{merge\_pull\_request}, but the head may move to $h_2$
before merge. GitHub's native API~\cite{githubmerge} can reject the merge if the
head no longer matches the reviewed commit, whereas the tested MCP
tool~\cite{githubmcp} exposes no equivalent check. Reads may all return $h_1$,
yet the head can change after the last one. Refusing or reading forever blocks
required work; proceeding without a condition may merge $h_2$. Thus, a component
limited to this MCP interface cannot guarantee both safety and completion,
although both interfaces behave as documented. The MCP interface reaches the
effect but cannot bind it to the reviewed commit.

We found a different failure in Kafka: a model can end too early. Our
initial model allowed return after every broker in the declared fixed set
applied the access control list (ACL) deletion~\cite{kafkakip801,kafkadeleteacls}. We
paused a request after authorization, yet it appended after return: authorizer
convergence was real, but the model stopped before the effect frontier at
append. We therefore designed Effect Fence: it blocks new use of revoked
authority and delays return until earlier in-flight work completes, extending
closure to append.
\emph{A correct proof can stop before the effect frontier}: verification is only
as meaningful as the effect path it reaches.

We model both failures along
$g\rightarrow a\rightarrow c\rightarrow F\rightarrow e$: grant lineage $g$
issues authorization instance $a$; continuation $c$ carries it; and effect $e$
follows the last preventable point $F$, the \emph{effect frontier}. The chain
distinguishes future-use closure, instance-effect closure, and quiescence (all
work stopping). We argue that authorization ends for a lineage only when no new
instance can be issued and no previously issued instance retains a path to a
rejected effect. We call the second condition \emph{policy-relative effect
closure}, or \emph{effect closure} for short. It is distinct from local
completion and quiescence: safe work may continue.

This is not only a check/use problem. In
NATS~\cite{natsdelivery,natssource} and Kafka, authorized work can still cause
effects after a local completion signal; Kubernetes~\cite{k8sadmission} adds
post-authorization mutation and already-admitted requests. The six paths span
late-bound names, post-check state changes, admitted execution, repeated
delivery, callbacks, and distributed in-flight work. Each shows how issued
authority can still cause an effect. We selected them for structural diversity rather than to estimate prevalence.

Prior work studies revocation and ongoing use, check/use and authorization--execution divergence, and enforcement over supplied interfaces
or models \cite{gligor1979revocation,uconabc2004,bishopdilger1996,mindthegap2025,
  authorizationexecutiongap2026,schneider2000}.
Commit-time authorization~\cite{cta2026} rechecks authority at an exposed commit
point.

In contrast, we present \system, which uses a finite contract constructed from provider evidence to decide
whether the provider boundary reaches the effect frontier and exposes enough
control to close rejected paths while required work completes. It returns a
strategy, a losing certificate, or \textsc{Unsupported} when evidence is
insufficient.
We contribute policy-relative effect-closure
semantics, frontier adequacy, and boundary realizability. Established
supervisory-control machinery~\cite{ramadgewonham1987,linwonham1988} decides the
finite realizability problem.
 Closure fails when a needed control is absent, local state hides active work, or the model stops before the frontier. Bind fixes the checked identity, Re-enter rechecks authority, and Fence waits at the effect boundary.

The NATS purge base and three variants exercise the closure checker. They hold
the provider path, policy, and required closure claim fixed while varying the
boundary assumption or control: native completion is treated as terminal, the
effect is bound to the authorized delivery, or dispatched callbacks are
awaited. Kafka provides the in-depth systems study: we implement a gate, test
crashes and contention, and measure whether it preserves unrelated work under
the declared policy.

  This paper makes three contributions:
  \begin{itemize}
    \item \textbf{Policy-relative effect-closure semantics.} We characterize
    authorization completion from grant to application effect through
    $g\to a\to c\to F\to e$. We define future-use, instance-effect, and lineage
    closure, prove their strict separation from quiescence, and recover complete
    revocation as lineage closure under a post-revocation policy. A closure claim returned on success names its kind, subject, policy, frontier, scope, and configuration.

    \item \textbf{Frontier adequacy and boundary realizability.}
    We require coverage from model boundary $B$ to effect frontier $F$, then ask
    whether the interface can report closure without blocking required work. For
    finite contracts, this is a control problem with hidden state. If no strategy works, a losing certificate rules out all strategies over that interface; its simplest form, an authorization twin, pairs worlds where the same observation requires incompatible choices.
    Machine-checked proofs establish the reduction and checker soundness. The
    checker validates 17 published certificates and derives every closure result
    for the four NATS contracts above. The checked results are conditional on the paths and continuations represented in each contract; they do not establish that every provider path was modeled.

  \item \textbf{Systems evidence.} Across six structural paths in GitHub,
  Kubernetes, NATS, and Kafka, we find missing controls, hidden active work, and
  premature model boundaries. We implement Effect Fence for the studied Kafka
  path and measure its cost in the test deployment.
  \end{itemize}

\section{Effect Closure at Provider Boundaries}
\label{sec:problem}
The \emph{mediator} observes events exposed by the provider and chooses among its exposed controls. Along a provider-to-effect path, these observations and
controls form the \emph{provider boundary}.
Late binding, transformation, redelivery, and execution share the path
$g\rightarrow a\rightarrow c\rightarrow F\rightarrow e$:
grant lineage $g$ issues authorization instance $a$, continuation $c$ carries it
across effect frontier $F$, and effect $e$ follows.
Instance $a$ represents authority for a principal, request, or work item;
$\mathsf{Permitted}(a,e)$ records which effects it authorizes. Causation alone
does not mean that $c$ carries $a$: evidence must link $c$'s creation or
continuation to $a$.
The contract must also represent effects that $c$ may cause outside
  $\mathsf{Permitted}(a,\cdot)$.
We follow this chain to the application's last chance to prevent a rejected
effect. Effect closure asks whether issued authority still has a path to such an
effect. Boundary realizability asks whether the exposed interface can remove
every such path while completing the required workload.

\subsection{From Calls to Policy-Relevant Effects}

\noindent\textbf{Provider paths and policy.}
A call at the provider boundary may continue internally to an effect. Interface
model $I$ specifies the mediator's observations and controls. Omitting an
exposed control can make realizability seem impossible; omitting a provider path
can make it seem possible.

An effect $e\in\mathcal E$ is a record of an operation and any policy-relevant
target, version, payload, principal, generation, or recipients. For an execution prefix $\rho_{\leq t}$, projection
$\eta(\rho_{\leq t})\in\mathcal E^*$ extracts its effect history. The application policy $\varphi_{\rm app}:\mathcal E^*\rightarrow\{0,1\}$ accepts or rejects such histories. Let $\mathsf{AuthCrossings}(\rho)$ be the pairs $(a,e)$ for which a continuation carrying $a$ produces $e$ in $\rho$. We
define prefix safety as
\[
\begin{aligned}
\mathsf{SafePrefix}(\rho)\equiv{}&
  \bigl(\varphi_{\rm app}(\eta(\rho))=1\bigr)\land{}\\[-1mm]
&\forall(a,e)\in\mathsf{AuthCrossings}(\rho),\\[-1mm]
&\hspace{17mm}\mathsf{Permitted}(a,e).
\end{aligned}
\]
For one effect $e$ reached by a continuation carrying $a$ after history $H$, the
corresponding one-step condition is
\[
\begin{aligned}
\mathsf{AuthorizedSafe}(a,H,e)\equiv{}&\mathsf{Permitted}(a,e)\\[-1mm]
&{}\land(\varphi_{\rm app}(H\cdot e)=1).
\end{aligned}
\]
The contract records the authorization instance carried by each continuation.
We call $e$ \emph{rejected} after $H$ if it lies outside that authorization or
$\varphi_{\rm app}$ rejects $H\cdot e$.

\noindent\textbf{The effect frontier.}
For effect $e$, its \emph{effect frontier} $F_e$ is the transition---or, in a
distributed execution, the cut---after which that effect can no longer be
prevented. It may precede return, declared completion, or visibility. Projection $\eta$
records $e$ at the crossing. A later policy-relevant event is a
separate effect with its own frontier. We write $F$ when the effect is clear.

The finite contract may stop at a proposed model boundary $B$. We write
$\mathsf{Adequate}_{\kappa}(B,F)$ when, within its declared scope $\kappa$, it represents
every path from $B$ by which issued authority can reach effect frontier $F$,
including every intervening preventive choice. Evidence must identify the last such choice
on every represented path, rule out a later one, and track the same effect
across model, checker, and trace. These checks can reject a proposed boundary but cannot
establish provider-wide path completeness. A later choice moves $F$; an omitted
in-scope path invalidates adequacy and requires extending the contract, as in
Kafka.

Because safety is checked at every prefix, later abort, compensation, or
reconciliation cannot erase a violation; prevention must act before $F$.

\noindent\textbf{Open continuations.}
An \emph{authorization-bearing continuation}, or carrier, is any state or
computation carrying issued instance $a$ toward an effect. Resolution,
transformation, queues, retries, callbacks, and in-flight execution create such
carriers. A carrier is \emph{open} if it can still reach a rejected effect before
the mediator has another chance to prevent it. \emph{Bind} preserves identity before later
resolution; \emph{Re-enter} rechecks authority after the last policy-relevant
change; and \emph{Fence} restricts or drains work already past that check. These
placements can be implemented in different ways. Updating future
authorization state alone does not close work already authorized.

\noindent\textbf{Observation alone is insufficient.}
Identical observations can require incompatible choices for safety and required
completion. An impossibility proof must include every exposed observation and
control: an omitted control might satisfy both requirements without
distinguishing the executions. We call such a pair an \emph{authorization twin}.
Figure~\ref{fig:boundary} shows the GitHub case: repeated reads cannot rule out
a later head change, and the base interface lacks Bind.

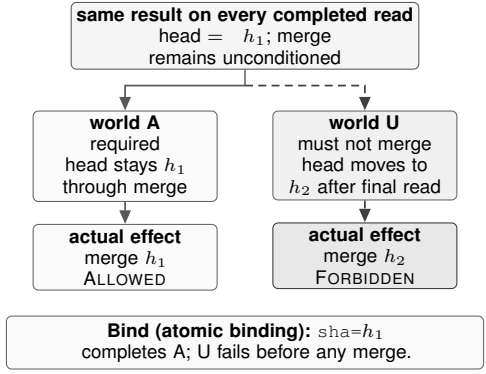
\begin{figure}[t]
\centering
\begin{tikzpicture}[
  font=\sffamily,
  same/.style={draw=black!65, rounded corners=2pt, fill=gray!8,
    text width=0.52\columnwidth, minimum height=7.3mm, align=center,
    inner sep=2pt, font=\scriptsize\sffamily},
  provider/.style={draw=black!60, rounded corners=2pt,
    text width=0.27\columnwidth, minimum height=7.8mm, align=center,
    inner sep=1.7pt, font=\scriptsize\sffamily},
  effect/.style={draw=black!60, rounded corners=2pt,
    text width=0.27\columnwidth, minimum height=7.8mm, align=center,
    inner sep=1.7pt, font=\scriptsize\sffamily},
  flow/.style={-{Triangle[length=1.8mm,width=1.5mm]}, line width=0.65pt},
  note/.style={draw=black!60, rounded corners=2pt, fill=black!2,
    text width=0.72\columnwidth, minimum height=7mm, align=center,
    inner sep=2pt, font=\scriptsize\sffamily}
]
\node[same] (view) {\textbf{same result on every completed read}\\
  head $=h_1$; merge remains unconditioned};
\node[provider, fill=black!2, below=5.4mm of view, xshift=-0.185\columnwidth] (pa)
  {\textbf{world A}\\required\\head stays $h_1$ through merge};
\node[provider, fill=black!6, below=5.4mm of view, xshift=0.185\columnwidth] (pu)
  {\textbf{world U}\\must not merge\\head moves to $h_2$ after final read};
\node[effect, fill=black!2, below=2.9mm of pa] (ea)
  {\textbf{actual effect}\\merge $h_1$\\\textsc{Allowed}};
\node[effect, fill=black!9, draw=black!78, below=2.9mm of pu] (eu)
  {\textbf{actual effect}\\merge $h_2$\\\textsc{Forbidden}};
\coordinate (fork) at ($(view.south)+(0,-2.1mm)$);
\draw[line width=0.65pt, black!55] (view.south) -- (fork);
\draw[flow, black!65] (fork) -| (pa.north);
\draw[flow, black!75, densely dashed] (fork) -| (pu.north);
\draw[flow, black!65] (pa.south) -- (ea.north);
\draw[flow, black!75, densely dashed] (pu.south) -- (eu.north);
\node[note, below=3.4mm of $(ea.south)!0.5!(eu.south)$] (bind)
  {\textbf{Bind (atomic binding):} \texttt{sha=$h_1$} \mbox{completes A}; U fails
   before any merge.};
\end{tikzpicture}
\caption{An authorization twin: both worlds return $h_1$ through the final read
but require incompatible choices for safety and completion; U may then move to
$h_2$. Repeated reading cannot complete A, whereas Bind supplies the reviewed SHA to the
native API's atomic check.}
\Description{Every completed read returns h1 in both worlds. In required world
A, the head stays h1 through merge. In unsafe world U, the head moves to
h2 after the mediator's final read, so an unconditioned merge is forbidden.
Reading forever blocks required completion; binding execution to h1 repairs the
boundary.}
\label{fig:boundary}
\end{figure}

\subsection{Future-Use Closure, Effect Closure, and Quiescence}

An information state $q$ contains every contract-admitted history consistent
with the mediator's observations. Future-use closure holds at $q$ if no later
in-scope attempt matching $g$ can issue another instance. Instance-effect closure
means that issued instance $a$ has no open carrier; instance quiescence means
that no carrier of $a$ can cross $F$ before the mediator can act again.
Figure~\ref{fig:closure-dimensions} compares the three conditions;
Section~\ref{sec:formal} defines them over finite analysis contracts.
Kafka shows that future-use closure can coexist with an open carrier of an
already-issued instance $a$ from grant lineage $g$:
\begin{equation*}
\mathsf{FutureUseClosed}(g,q)
\not\Rightarrow
\mathsf{EffectClosed}_{\varphi_{\rm app}}(a,q).
\end{equation*}
After every authorizer in the fixed broker set has applied the ACL deletion, an
earlier request can still append. Instance-effect closure likewise does not
guarantee quiescence:
\begin{equation*}
\mathsf{EffectClosed}_{\varphi_{\rm app}}(a,q)
\not\Rightarrow
\mathsf{Quiescent}(a,q).
\end{equation*}
Bind may leave the merge running because its SHA condition permits
only the authorized, policy-accepted outcome. Commit-time
authorization~\cite{cta2026} can implement Re-enter at commit. Thus quiescence
guarantees instance-effect closure but is unnecessary: both allow a carrier to
cross $F$ safely.

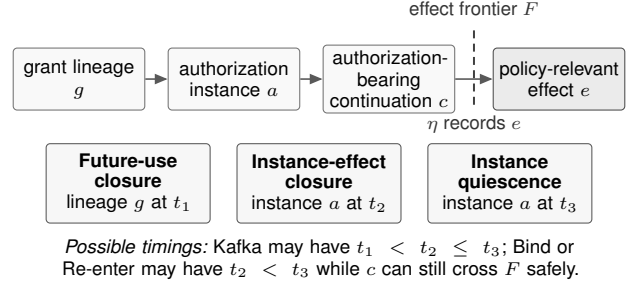
\begin{figure}[t]
\centering
\begin{tikzpicture}[
  font=\scriptsize\sffamily,
  event/.style={draw=black!55, rounded corners=1.5pt, fill=gray!5,
    align=center, minimum height=8.2mm, text width=0.185\columnwidth,
    inner sep=2.2pt},
  flow/.style={-{Triangle[length=1.7mm,width=1.35mm]}, semithick,
    draw=black!65},
  condition/.style={draw=black!65, rounded corners=1.5pt, fill=black!3,
    align=center, text width=0.235\columnwidth, minimum height=10.5mm,
    inner sep=2.2pt},
  note/.style={font=\scriptsize\sffamily, align=center,
    text width=0.94\columnwidth},
  frontier/.style={densely dashed, draw=black!75, line width=0.75pt}
]
\path[use as bounding box] (-1.00,1.10) rectangle (7.15,-2.65);
\node[event] (grant) at (-0.10,0) {grant lineage\\$g$};
\node[event] (auth) at (1.95,0) {authorization\\instance $a$};
\node[event] (cont) at (4.00,0) {authorization-bearing\\continuation $c$};
\node[event, fill=black!8, draw=black!70] (effect) at (6.25,0)
  {\mbox{policy-relevant}\\effect $e$};
\coordinate (effectfrontier) at ($(cont.east)!0.50!(effect.west)$);
\draw[frontier] ([yshift=-3.4mm]effectfrontier) --
  ([yshift=7.0mm]effectfrontier);
\node[font=\scriptsize\sffamily, text=black!80, anchor=south]
  at ([yshift=7.2mm]effectfrontier) {effect frontier $F$};
\node[font=\scriptsize\sffamily, text=black!80, anchor=north]
  at ([yshift=-3.6mm]effectfrontier) {$\eta$ records $e$};
\draw[flow] (grant) -- (auth);
\draw[flow] (auth) -- (cont);
\draw[flow] (cont) -- (effect);

\node[condition] (t1) at (0.55,-1.38)
  {\textbf{Future-use}\\[-0.2mm]\textbf{closure}\\lineage $g$ at $t_1$};
\node[condition] (t2) at (3.075,-1.38)
  {\textbf{Instance-effect}\\[-0.2mm]\textbf{closure}\\instance $a$ at $t_2$};
\node[condition] (t3) at (5.60,-1.38)
  {\textbf{Instance}\\[-0.2mm]\textbf{quiescence}\\instance $a$ at $t_3$};
\node[note] at (3.075,-2.38)
  {\emph{Possible timings:} Kafka may have
   $t_1<t_2\leq t_3$; Bind or \mbox{Re-enter} may have $t_2<t_3$ while $c$
   can still cross $F$ safely.};
\end{tikzpicture}
\caption{Future-use closure, instance-effect closure, and quiescence need not
coincide. Projection $\eta$ records effect $e$ at $F$; provider completion is
separate and may occur before or after that crossing.}
\Description{Grant lineage g issues authorization instance a, which the authorization-bearing
continuation c carries toward a policy-relevant effect across frontier F. The
conditions are distinct: quiescence implies instance-effect closure, but
future-use closure need not imply instance-effect closure, and instance-effect
closure need not imply quiescence. Example timings show Kafka with t1 before t2
at or before t3, and Bind or Re-enter with effect closure before quiescence while
continuation c can still cross F safely.}
\label{fig:closure-dimensions}
\end{figure}

\subsection{Security Objective and Claim Boundary}

A closure claim must state its subject and kind: revocation may close
future use and already-issued instances; an operation may close one instance;
authorization acceptance may promise neither. None by itself asserts quiescence.

\noindent\textbf{Covered worlds, required worlds, and public claims.}
A world is a bounded, in-scope scenario: a request, authorization choices, a
declared response, and a completion goal. $W_P$ contains the covered worlds,
each required to reach a declared response within the bound.
$W_A\subseteq W_P$ marks those in which authorized work must also succeed, and
$\Xi(w)$ lists the public claims due then. Scope $\kappa$ records
the provider and interface version,
policy, workload, effect frontier, permitted hooks, credentials and controls,
evidence, and exclusions. Configuration version $\beta$ identifies one
installed mediator strategy $M$. Every
closure claim carries $\beta$ because enforcement depends on that configuration.
Write $\mathbf W=(W_P,W_A,\Xi)$ for these worlds and claim requirements.

\begin{effectcontract}
The analysis contract defines three public claim kinds, each scoped to
$(\kappa,\beta)$.

\noindent\textbf{Future use.}
$\mathsf{FutureUseClosed}(g)@(\kappa,\beta)$: under $\beta$, no admitted future
attempt matching $g$ can issue a derived instance.

\noindent\textbf{Issued instance.}
$\mathsf{InstanceEffectClosed}(a)@(\kappa,\beta)$ states that, under $\beta$, no
in-scope carrier of $a$ admitted by the contract can cross $F$ to an effect that
fails $\mathsf{AuthorizedSafe}$.

\noindent\textbf{Lineage.}
$\mathsf{LineageClosed}(g)@(\kappa,\beta)$ combines future-use closure with
effect closure under $\beta$ for every in-scope instance issued from $g$.

\noindent\textbf{Completion and evidence.}
$\Xi(w)$ lists required claims. Every emitted claim is tied to $(\kappa,\beta)$.
Interface model $I$ lists observations and controls; versioned evidence supports
provider premises and exclusions.
\end{effectcontract}

Let $P$ be contract-admitted provider behavior and $M$ a mediator.
$\mathsf{Runs}_P(w,M)$ contains every admitted run in $w$, and
$\mathsf{Runs}_{P,W_P}(M)$ their union over $W_P$.
Write $\mathsf{ClaimsDone}_{\Xi}(w,\rho)$ when, for every
$\bar\chi\in\Xi(w)$, $\rho$ emits a claim $\chi$ satisfying
$\mathsf{Match}_{\kappa,\beta}(\bar\chi,\chi)$; Section~\ref{sec:formal} also
checks the claim's truth.
$\mathsf{RequiredDone}_{\Xi}(w,\rho)$ holds when the world's declared completion
goal has been reached and those claims have been emitted. Then
\begin{align*}
\mathsf{Completes}_P(w,M)
  &\iff \forall\rho\in\mathsf{Runs}_P(w,M):\\[-1mm]
  &\qquad \Diamond\mathsf{RequiredDone}_{\Xi}(w,\rho).
\end{align*}
Appendix~\ref{app:finite-detail} defines exact matching. The schedule bound
limits how long an admitted run may take to reach its declared response and
enter the checked post-settlement invariant; it does not limit a claim's lifetime.
Write $\contractsettled_{W_P}(M)$ for this property. Machine checking verifies
entry, preservation, policy safety, and claim safety; evidence must cover all
in-scope successors. Settlement requires neither carrier nor provider quiescence.

\noindent\textbf{Admitted behavior.}
Provider and environment may choose any transition, retry, interleaving, or
delay allowed by $P$ and $\mathbf W$. The guarantee covers malicious or benign
callers and every admitted run, but not behavior outside the declared contract.

Every analyzed contract has $W_A\neq\varnothing$. Safety alone permits a useless
deny-all mediator, but required completion does not. We require both
\begin{equation}
\begin{aligned}
\mathsf{PolicySafe}_{W_P}(M):\quad
&\forall\rho\in\mathsf{Runs}_{P,W_P}(M),\\[-1mm]
&\forall t\leq |\rho|:\ \mathsf{SafePrefix}(\rho_{\leq t}),
\end{aligned}
\end{equation}
and
\begin{equation}
\begin{aligned}
&\requiredcomplete_{W_A}(M):\\[-1mm]
&\quad \forall w\in W_A:\ \mathsf{Completes}_P(w,M).
\end{aligned}
\end{equation}
Together with truthful claims and contract settlement, they exclude forbidden
prefixes, false closure returns, a deny-all mediator, and covered worlds that
never return.

\noindent\textbf{Security boundary and trusted computing base.}
Correctness assumes every covered invocation passes through the mediator and
interface $I$ with no alternate credential or channel, repairs follow modeled
ordering, and nondeterministic provider execution conforms to the finite
analysis contract.
Policy and workload are inputs; correctness does not rely on a safe caller.

\noindent\textbf{Scope and exclusions.}
A claim covers only contractual scope $\kappa$ of the finite analysis contract
$\langle P,I,\varphi_{\rm app},\mathbf W\rangle$. It excludes Byzantine behavior,
enforcement compromise, unmodeled channels, and availability beyond $\mathbf W$.
If evidence does not support $P$ or $I$, \system returns \textsc{Unsupported};
no certificate extends this scope to the provider as a whole.

\section{Provider-Boundary Realizability}
\label{sec:formal}

Starting from the finite, bounded analysis contract of
Section~\ref{sec:problem}, assumed exact within its stated scope, we reduce its
closure objective to a partially observed discrete-event system (DES): mediator
actions are controllable, provider and environment moves are not, and $I$
determines what the mediator observes~\cite{ramadgewonham1987,linwonham1988}.

The reduction determines whether a mediator using only observations in $I$ can
keep policy and claims safe, reach the checked post-settlement invariant, and
complete required work. A winning certificate gives such a strategy whereas a losing certificate rules out all such strategies in the finite contract. Repairs may
neither block required work nor change claims, workload, or effect semantics.

\subsection{Finite Boundary Model}

We model a provider and its exposed interface as
\[
\begin{gathered}
P=\langle S,\mathsf{Init}_P,\mathcal A,\delta,\lambda,\gamma,\eta\rangle,
\quad \mathcal A=A_R\uplus A_H,\\[-1mm]
I=\langle\pi_I,\Gamma_I\rangle,\\[-1mm]
\mathcal C=\langle P,I,\varphi_{\rm app},\mathbf W\rangle,
\quad \mathbf W=(W_P,W_A,\Xi).
\end{gathered}
\]
Here $S$ is hidden state, $\mathsf{Init}_P(w)$ gives world $w$'s admitted initial
states, and $\mathcal A=A_R\uplus A_H$ separates mediator controls from hidden
provider or environment moves. Relation $\delta$ gives transitions and
$\lambda$ labels them. The mediator observes $\pi_I\!\circ\lambda$ and chooses
an allowed control from $\Gamma_I$; abort and cancellation have explicit outcomes. Map $\gamma$
marks frontier crossings, $\eta$ projects effect histories, and $\mathsf{SafePrefix}$
also checks whether each effect is permitted by the authorization instance
carried by the continuation that produced it. Workload $\mathbf W$ gives covered
worlds, required-success worlds, and claim templates, while cope $\kappa$ records where they apply.

The model makes control and uncertainty explicit but does not infer provider
routes or enlarge $\kappa$. The verdict remains within the
contract's scope.

\noindent\textbf{Information states.}
For observed history $o$, information state $q=q_I(o)$ contains every
contract-admitted history consistent with $o$. Write $\Gamma(q)$ for the controls
shared by every history in $q$, and $q_0$ for the initial state. After the
mediator chooses $u\in\Gamma(q)$, hidden moves continue until the mediator can
act again or the bounded run ends. $\mathsf{Out}(q,u)$ contains every
admitted endpoint. The construction separately records whether an intermediate
prefix violates safety.
Appendix~\ref{app:finite-detail} defines these segments, their coverage, and the
exact grouping of histories into information states.

The provider may therefore choose the worst admitted outcome: a control is safe
only if every admitted outcome is safe. Bounds restrict the contract being
analyzed; they never justify discarding an admitted outcome within it.

Exactness is two-sided: omitting hidden behavior can create a false winner, while
omitting an exposed control can create a false impossibility. If evidence cannot
justify both directions, the analysis returns \textsc{Unsupported} before solving.

\subsection{Interpreting Typed Closure Claims}

\noindent\textbf{Provenance and closure.}
We now interpret Section~\ref{sec:problem}'s public claims over information
states. For the path $g\rightarrow a\rightarrow c\rightarrow F\rightarrow e$,
$\mathsf{DerivedFrom}(a,g)$ says that $a$ may derive from $g$;
$\mathsf{IssuedFrom}(a,g,q)$ also requires $a$ to appear in $q$'s issuance record; and
$\mathsf{Carries}(c,a,H,q)$ says live continuation $c$ carries $a$ with effect
history $H$. Each represented link must cite evidence, and a carrier link must connect
$c$'s creation or continuation to $a$ and record what $a$ authorized.
Physical outcomes are recorded
separately and may fall outside $\mathsf{Permitted}(a,\cdot)$. Excluding such
outcomes would remove the GitHub moved-head path.

Future-use closure covers every in-scope future attempt matched to $g$ rather than only
choices made by the installed mediator. Instance-effect closure covers the histories in
$q$ and their live carriers: before the mediator can act again, no carrier may
cross $F$ to an effect that fails $\mathsf{AuthorizedSafe}$. Carriers may remain
if every such crossing is safe;
quiescence excludes every carrier that can still cross $F$.
Appendix~\ref{app:finite-detail} defines the predicates;
$\mathsf{EffectClosed}_{\varphi_{\rm app}}(a,q)$ interprets the public
$\mathsf{InstanceEffectClosed}(a)$ claim.

At fixed scope and configuration, a public lineage claim is true at $q$
exactly when
\begin{equation}
\begin{aligned}
&\mathsf{LineageClosed}_{\varphi_{\rm app}}(g,q)\\[-1mm]
&\quad\iff\mathsf{FutureUseClosed}(g,q)\land{}\\[-1mm]
&\qquad\forall a:\ \mathsf{IssuedFrom}(a,g,q)\Rightarrow{}\\[-1mm]
&\hspace{26mm}\mathsf{EffectClosed}_{\varphi_{\rm app}}(a,q).
\end{aligned}
\label{eq:lineage-closure}
\end{equation}
$\mathsf{LineageEffectClosed}_{\varphi_{\rm app}}(g,q)$ abbreviates the
condition that every instance issued from $g$ is effect-closed.

\noindent\textbf{Claim scope and configuration.}
A claim's subject and kind determine its meaning. Future-use closure covers later
issuance, instance-effect closure covers already-issued instances, and lineage
closure combines both. None of these requires quiescence. Instance-effect and lineage
closure allow carriers whose reachable effects satisfy $\mathsf{AuthorizedSafe}$.

Each emitted claim instantiates its template with $(\kappa,\beta)$, where $\beta$
identifies installed strategy $M$, and must match exactly. The claim remains tied
to that configuration after return. Because the model retains earlier issuance
and carriers, switching to $(\beta',M')$ must preserve every prior claim. Each
checked contract covers one $\beta$, but preserving claims across configuration
changes is a separate deployment obligation.

\begin{proposition}[Closure separations]
\label{prop:closure-dimensions}
Over finite well-formed provider-boundary contracts,
$\mathsf{Quiescent}(a,q)\Rightarrow
\mathsf{EffectClosed}_{\varphi_{\rm app}}(a,q)$, whereas
\begin{equation}
\begin{gathered}
\mathsf{FutureUseClosed}(g,q)
\not\Rightarrow\mathsf{LineageEffectClosed}_{\varphi_{\rm app}}(g,q),\\[-1mm]
\mathsf{LineageEffectClosed}_{\varphi_{\rm app}}(g,q)
\not\Rightarrow\mathsf{FutureUseClosed}(g,q),\\[-1mm]
\mathsf{EffectClosed}_{\varphi_{\rm app}}(a,q)
\not\Rightarrow\mathsf{Quiescent}(a,q).
\end{gathered}
\label{eq:closure-dimensions}
\end{equation}
\end{proposition}

Under a post-revocation application policy $\varphi_{\rm app}^{\rm rev}$ that
rejects every later use of revoked authority, lineage closure corresponds to
complete revocation in the classical model of migrated permissions~\cite{flask1999}.

\subsection{Deciding Realizability}

$\mathsf{PolicySafe}(q)$ requires every represented execution prefix to satisfy
$\mathsf{SafePrefix}$. $\mathsf{SuccessObserved}(\chi,q)$ records a visible
success response carrying claim $\chi$, and
$\mathcal I_{\varphi_{\rm app},\kappa}(\chi,q)$ gives its truth condition at $q$.
Visible closure claims are safe exactly when
\begin{equation}
\begin{aligned}
&\mathsf{ClaimSafe}(q)\\[-1mm]
&\quad\iff\forall\chi:\ \mathsf{SuccessObserved}(\chi,q)\Rightarrow\\[-1mm]
&\hspace{32mm}\mathcal I_{\varphi_{\rm app},\kappa}(\chi,q).
\end{aligned}
\label{eq:claim-safe}
\end{equation}
Write $\mathsf{PolicySafe}_{W_P}(M)$ for policy-safety invariance and
$\mathsf{ClaimsSafe}_{W_P}(M)$ for claim-safety invariance under $M$ from
$W_P$. Predicate $\contractsettled_{W_P}(M)$ requires every admitted run to
reach a declared response and enter a checked post-settlement invariant within
the bound. $\requiredcomplete_{W_A}(M)$ additionally requires each $W_A$ world
to reach its goal and emit all required claims. Contract settlement ends the
bounded wait for a response, and it does not require carrier or provider quiescence.

\begin{definition}[Provider-boundary realizability]
A contract $\mathcal C=\langle P,I,\varphi_{\rm app},\mathbf W\rangle$ is realizable iff
\begin{equation}
\begin{aligned}
&\realizable(P,I,\varphi_{\rm app},\mathbf W)\\[-1mm]
&\quad\iff\exists M:\quad
(\forall o,\ M(o)\in\Gamma_I(o))\\[-1mm]
&\hspace{17mm}\land\mathsf{PolicySafe}_{W_P}(M)\\[-1mm]
&\hspace{17mm}\land\mathsf{ClaimsSafe}_{W_P}(M)\\[-1mm]
&\hspace{17mm}\land\contractsettled_{W_P}(M)\\[-1mm]
&\hspace{17mm}\land\requiredcomplete_{W_A}(M).
\end{aligned}
\label{eq:realizability}
\end{equation}
\end{definition}

\noindent\textbf{Frontier adequacy.}
Boundary power matters only if the model reaches the effect. A positive verdict
at boundary $B$ may fail if an omitted path lets an authorization-bearing
continuation reach $F$; one carrier can introduce a rejected path. Kafka's
authorizer-convergence model stopped before append to the target leader's
local log. Extending it to $F_{\rm append}$ invalidates the verdict, so the
evidence check rejects the premature contract before solving.
$\mathsf{Adequate}_{\kappa}(B,F)$ remains an obligation within scope $\kappa$.

\noindent\textbf{Boundary insufficiency.}
Even at the effect frontier, the exposed observations and controls may be
insufficient. An \emph{authorization twin} is the boundary
version of the standard DES indistinguishability argument. Write
$\mathsf{Twin}_I(w_A,w_U)$ when, against every mediator limited to $I$, the
provider can keep required world $w_A$ and unsafe world $w_U$ observationally
equivalent until required completion becomes impossible in $w_A$ or a forbidden
prefix occurs in $w_U$.
\begin{corollary}[Authorization-twin obstruction]
If $w_A,w_U\in W_P$, $w_A\in W_A$, and
$\mathsf{Twin}_I(w_A,w_U)$, no mediator $M$ restricted to $I$ satisfies both
$\mathsf{PolicySafe}_{W_P}(M)$ and $\requiredcomplete_{W_A}(M)$.
\end{corollary}
The obstruction concerns boundary power rather than mediator complexity. Breaking the
twin requires different provider behavior or another observation or control. A
twin is sufficient but not necessary for unrealizability, because some losing
information states contain more than two worlds. Appendix~\ref{app:finite-detail}
treats randomization and gives the general certificate checks.

The GitHub case exposes one limit. Repeated head reads cannot permanently
distinguish the two worlds because the provider may move the head after the
final read.
Merging may then be unsafe, whereas aborting or reading indefinitely violates
required completion. Making the checked version an atomic precondition of the
effect breaks this indistinguishability. Once the full path is represented,
searching harder over the unchanged boundary cannot produce a valid solution.

\noindent\textbf{Finite reduction.}
With the claim and frontier fixed, $q\in\mathsf{Pre}(X)$ when some available
control keeps every admitted outcome in $X$. Call $q$ admissible when policy and claims are safe:
$\mathsf{Admissible}(q)\equiv\mathsf{PolicySafe}(q)\land\mathsf{ClaimSafe}(q)$.
Admissible means safe; admitted means included by the contract. Here
$\mathsf{ContractSettled}(q)$ means that every represented state has emitted
its declared response and belongs to the checked post-settlement invariant;
$\mathsf{RequiredDone}_{\Xi}(q)$ records each designated goal
and required template. Let
\begin{equation*}
\begin{aligned}
G=\{q\mid{}&\mathsf{Admissible}(q)\land\mathsf{ContractSettled}(q)\\[-1mm]
&{}\land\mathsf{RequiredDone}_{\Xi}(q)\}.
\end{aligned}
\end{equation*}
Here, $\mathsf{RequiredDone}_{\Xi}$ includes a matching observed claim for
every template required in each represented required world. The ranked least
fixed point for reaching $G$ is
\begin{equation}
\begin{aligned}
\mathsf{Pre}(X)=\{q\mid{}&\mathsf{Admissible}(q)\land
\exists u\in\Gamma(q):\\[-1mm]
&\varnothing\ne\mathsf{Out}(q,u)\subseteq X\},\\[-1mm]
\mathsf{Win}_0&=G,\qquad
\mathsf{Win}_{n+1}=G\cup\mathsf{Pre}(\mathsf{Win}_n),\\[-1mm]
\mathcal W&=\bigcup_n\mathsf{Win}_n.
\end{aligned}
\label{eq:winning}
\end{equation}
Set $\mathcal W$ is the \emph{winning region}: from every state in $\mathcal W$,
the mediator can force $G$. The subset condition covers every
admitted provider outcome; decreasing rank proves entry into the post-settlement
invariant and required completion. Lean~4~\cite{demouraullrich2021lean4}
checks that the admitted post-settlement successor relation preserves the invariant and
$\mathsf{Admissible}$. A settled state may retain work with no rejected effect path;
no quiescence is required.

\begin{theorem}[Finite boundary-realizability reduction]
\label{thm:realizability}
For a finite, bounded provider-boundary contract whose information-state construction
is well formed and exact, the construction above yields a finite partially observed DES control
instance with information-state winning region $\mathcal W$, and
\[
\realizable(P,I,\varphi_{\rm app},\mathbf W)\quad\Longleftrightarrow\quad q_0\in\mathcal W.
\]
\end{theorem}

Lean~4 machine-checks the equivalence above,
certificate soundness, and each translation from the machine-readable contract
to the game; the contract and its in-scope route inventory remain empirical.

\begin{corollary}[Success honors its closure claims]
\label{cor:continuation-closure}
Every state reachable under a winning strategy is admissible, so each observed
success satisfies every future-use, instance-effect, or lineage claim that it
declares. A certificate guarantees only the closure properties that the corresponding success response
advertises, and effect closure need not imply quiescence.
\end{corollary}

\begin{figure*}[t]
\centering
\resizebox{0.99\textwidth}{!}{%
\begin{tikzpicture}[
  font=\sffamily,
  stage/.style={draw=black!60, rounded corners=2pt, align=center,
    text width=2.08cm, minimum height=12mm, inner sep=2pt,
    font=\scriptsize\sffamily},
  lower/.style={draw=black!60, rounded corners=2pt, align=center,
    text width=1.98cm, minimum height=10mm, inner sep=2pt,
    font=\scriptsize\sffamily},
  flow/.style={-{Triangle[length=1.8mm,width=1.5mm]}, line width=0.65pt},
  branch/.style={-{Triangle[length=1.6mm,width=1.3mm]}, dashed,
    draw=black!58, line width=0.55pt}
]
\node[stage, fill=gray!8] (provider)
  {\textbf{provider evidence}\\docs, source, traces};
\node[stage, fill=gray!5, anchor=north west] (contract)
  at ([xshift=4mm]provider.north east)
  {\textbf{grounded contract}\\machine checking starts here};
\node[stage, fill=gray!8, anchor=north west] (compiler)
  at ([xshift=7mm]contract.north east)
  {\textbf{Lean translation validator}\\checks proposed game};
\node[stage, fill=gray!5, anchor=north west] (artifact)
  at ([xshift=4mm]compiler.north east)
  {\textbf{untrusted solver}\\proposes certificate};
\node[stage, fill=gray!8, anchor=north west] (checker)
  at ([xshift=7mm]artifact.north east)
  {\textbf{Lean certificate checker}\\checks certificate};
\node[stage, fill=gray!5, anchor=north west] (claim)
  at ([xshift=4mm]checker.north east)
  {\textbf{checked finite verdict}\\scoped to the contract};
\draw[flow] (provider) -- (contract);
\draw[flow] (contract) -- (compiler);
\draw[flow] (compiler) -- (artifact);
\draw[flow] (artifact) -- (checker);
\draw[flow] (checker) -- (claim);

\coordinate (lowerline) at ([yshift=-25mm]provider.north);
\node[lower, fill=gray!10, anchor=north] (unsupported)
  at (contract.center |- lowerline)
  {\textbf{Unsupported}\\missing premise};
\draw[branch] (contract) -- (unsupported);

\node[lower, fill=gray!5, anchor=north] (basis)
  at (compiler.center |- lowerline)
  {fixed repair\\set};
\node[lower, fill=gray!8, anchor=north] (patched)
  at (artifact.center |- lowerline)
  {checked repaired\\contract $\mathcal C\oplus\Delta$};
\node[lower, fill=gray!5, anchor=north] (obligation)
  at (checker.center |- lowerline)
  {deployment\\obligation};
\node[lower, fill=gray!8, anchor=north] (trace)
  at (claim.center |- lowerline)
  {provider trace\\conformance};
\coordinate (lossbranch) at ([yshift=-5mm]claim.south);
\coordinate (lossleft) at (lossbranch -| basis.north);
\draw[branch] (claim.south) -- (lossbranch) --
  node[pos=0.15, below=0.5mm, inner sep=0pt,
    font=\scriptsize\sffamily] {if losing}
  (lossleft) -- (basis.north);
\draw[flow] (basis) -- (patched);
\draw[flow] (patched) -- (obligation);
\draw[flow] (obligation) -- (trace);

\begin{pgfonlayer}{background}
\node[draw=black!35, dashed, rounded corners=3pt,
  fit=(provider)(contract), inner xsep=1.5mm, inner ysep=2mm,
  label={[font=\scriptsize\sffamily\bfseries,yshift=1mm]above:empirical grounding}] {};
\node[draw=black!35, dashed, rounded corners=3pt,
  fit=(compiler)(artifact)(checker), inner xsep=1.5mm, inner ysep=2mm,
  label={[font=\scriptsize\sffamily\bfseries,yshift=1mm]above:mechanized decision}] {};
\node[draw=black!50, dashed, rounded corners=3pt,
  fit=(claim), inner xsep=1.5mm, inner ysep=2mm,
  label={[font=\scriptsize\sffamily\bfseries,yshift=1mm]above:scoped verdict}] {};
\node[draw=black!35, dashed, rounded corners=3pt,
  fit=(basis)(patched)(obligation)(trace), inner xsep=1.5mm, inner ysep=2mm,
  label={[font=\scriptsize\sffamily\bfseries,yshift=1mm]above:
    repair and deployment}] {};
\end{pgfonlayer}
\end{tikzpicture}%
}
\caption{The pipeline separates empirical grounding, machine checking, and the
scoped verdict. Evidence supports a finite contract or yields
\textsc{Unsupported}. Lean validates the proposed game and checks the untrusted
solver's certificate. A losing verdict may enter repair selection; deployment
obligations and trace conformance remain separate. Provider meaning and coverage
still come from the evidence.}
\Description{A left-to-right pipeline has three zones: empirical grounding,
mechanized decision, and a scoped verdict. Provider evidence yields either a
finite semantic contract or \textsc{Unsupported}. Lean validates the proposed game and
checks the solver's certificate. A losing verdict branches to candidate repair,
deployment obligations, and provider-trace conformance.}
\label{fig:pipeline}
\end{figure*}

\noindent\textbf{Checked translation and certificates.}
Following translation validation~\cite{pnueli1998translation}, an untrusted
compiler may propose a finite game. From the machine-readable contract, Lean
reconstructs its materialized safety labels, goals, actions, and successors,
rejecting duplicate or unreachable information states. A winning certificate
must give a safe ranked region with decreasing successors. For a losing set
$L$, the checker verifies $q_0\in L$, every $q\in L$ lies outside $G$, and, at
each admissible $q\in L$, every available action has an admitted successor in
$L$. Thus, neither compiler nor solver is trusted after the contract is fixed.

The certificates establish different facts. A winning certificate selects
actions that induce one strategy over $I$. That strategy preserves safety, makes
only true claims, and completes required work. A losing certificate gives, for every available action at each
admissible $q\in L$, an admitted successor in $L$, while a non-admissible state already
violates safety, so it needs no successor witness. It rules out every mediator
within the finite contract and omitted provider routes remain outside the result.

\subsection{Compatible Boundary Changes}

\noindent\textbf{Preserving external meaning.}
A losing verdict motivates a repair only if the changed boundary preserves the
same problem. Candidate $\Delta$ may change modeled provider behavior or the
interface, but not the policy or external specification. It yields $\mathcal C\oplus\Delta=
\langle P_\Delta,I_\Delta,\varphi_{\rm app},\mathbf W_\Delta\rangle$ and a mapping
$\iota_\Delta:W_P\to W_P^\Delta$.
$\mathsf{PreservesExternalContract}$ requires $\iota_\Delta$ to be a bijection
preserving the request, principal, initial authority, authorization choices,
completion goal, $W_A$ membership, and the claim template attached to each
response. In particular, $\Xi_\Delta(\iota_\Delta(w))=\Xi(w)$ for every $w\in W_A$.

A candidate remains within $\kappa$ only while it stays inside the recorded
contract boundary. It may refine internal state, install an eligible control or
new $\beta_\Delta$, and delay success within the declared schedule bound, but it
must re-establish the same public templates. The checker requires an exact
match; one claim kind cannot replace another. Stronger enforcement is allowed,
but a different public claim defines a different contract.
The candidate must also preserve the projected effects and their meaning under
$\eta$. It may add internal state or evidence and make an effect unreachable,
but cannot erase a projected field, relabel an external consequence, or change
how $\varphi_{\rm app}$ interprets it. Reachability may change, but effect meaning may not.

These restrictions prevent redefinition: workload preservation excludes
refusal, claim-template preservation excludes changing the public contract, and
effect preservation excludes hiding a forbidden consequence. A new observation
or control may change the preventive choices leading to $F$, but not the
external consequence judged by $\varphi_{\rm app}$.

\noindent\textbf{Search and minimality.}
Given a declared repair set $\mathcal B$ and cost order, \system evaluates the
combinations that $\mathcal B$ permits and reports each Pareto-minimal repair
(one not dominated under that order), or
\textsc{NoRepairFound} when the declared repair set contains no winning candidate.
Minimality is relative to the finite contract, repair set,
bounds, and declared cost order. The set may contain the Bind, Re-enter, and
Fence placements, as well as other declared changes. \system then validates
these candidates and connects the repaired contract to deployment traces.

\section{\texorpdfstring{\textsc{EffectBound}}{EFFECTBOUND}}
\label{sec:system}

\system proceeds in three stages (Figure~\ref{fig:pipeline}): construct an
analysis contract from provider evidence, check its verdict and any repair, and
validate deployment traces.

\subsection{Grounding the Contract}

\noindent\textbf{Contract inputs.}
A contract records provider behavior and interface, policy and workload, effect
frontier, authorization origins and continuation links, claims, exclusions, and
repair candidates. Schemas alone do not determine when targets resolve, whether
operations are atomic, or how retries behave or requests fan out. Documentation, pinned
source, controlled execution, or instrumentation must support those details.
\system checks the resulting contract, but it does not discover provider
paths.

\noindent\textbf{Tracing authorization to effects.}
For every represented path, the analyst follows $g\to a\to c$ toward a
proposed model boundary $B$, then asks whether any carrier can reach the
evidence-supported effect frontier $F$. The analyst asks where authority is bound,
whether an earlier check remains valid, and whether in-flight work can continue.
Tests exercise late resolution, mutation, queues, redelivery, restart, and
membership changes. The evidence record identifies the source and hash for each
path, exclusion, frontier, and response meaning. Validation rejects missing
evidence or hash mismatches.

\noindent\textbf{Fail-closed validation.}
A response either makes a closure claim or abstains from making one. \textsc{Unsupported}
results when authorization origins, continuation links, or declared exclusions
lack evidence, or when coverage of hidden behavior or exposed controls is
incomplete. The Kafka trace for one authorized request continues beyond the
proposed Kafka convergence boundary $B_{\rm applied}$ to the append frontier
$F_{\rm append}$, so grounding rejects that contract before solving.
Validation fails closed when represented premises lack support. It cannot detect
a provider path absent from both the evidence inventory and the contract.

\subsection{Checking Verdicts and Repairs}

\noindent\textbf{Verdicts and assurance.}
A grounded contract yields a checked winning strategy or losing certificate.
Lean validates all 17 published source-to-game translations and finite-game
certificates. Four NATS purge contracts---base, Bind, native-terminal (zero
counters treated as terminal), and Tracked Drain (awaiting tracked callback
completion)---declare the same instance-effect-closure requirement. For these four,
Lean checks carrier-link well-formedness and derives each result from the
materialized issuance, carrier, frontier-crossing suffix, permission,
application-policy, and coverage relations.
The compiler and solver need not be trusted after the finite contract is constructed
(Table~\ref{tab:assurance-boundary}).

\begin{table}[t]
\centering
\caption{Assurance boundary. Lean consumes, but does not infer, provider
semantics.}
\label{tab:assurance-boundary}
\footnotesize
\setlength{\tabcolsep}{3pt}
\renewcommand{\tabularxcolumn}[1]{m{#1}}
\begin{tabularx}{\columnwidth}{@{}m{0.55\columnwidth}X@{}}
\toprule
Material & Status \\
\midrule
Policy, authorization scope, workload, claim kind/subject, exclusions, repair set & Analyst-authored \\
Boundary/frontier; authorization links/effect paths; recorded safety/policy verdicts; scope/configuration match flags; scoped issuance/carrier coverage; post-settlement relation/invariant & Grounded/materialized input \\
Carrier-link well-formedness; invariant entry, self-loops, inductiveness, claim persistence, and safety; closure, game translation, certificate & Lean-checked \\
Provider completeness; conformance of all provider executions & Not established \\
\bottomrule
\end{tabularx}
\end{table}

\noindent\textbf{Why analysis precedes repair.}
Testing one Fence cannot decide whether another strategy over $I$ works. Safety,
truthful claims, and completion rule out denying everything. A losing certificate
has $q_0\in L\subseteq Q\setminus G$. Every action available at an admissible
$q\in L$ has an admitted successor in $L$.
Grounding yields \textsc{Unsupported} when required premises lack evidence.
Candidate checking rejects changes that block required work, reuse an
outdated configuration version, or alter claims or effect semantics. The analysis thus separates
insufficient grounding from missing boundary power before checking Bind,
Re-enter, or Fence candidates from the declared repair set.

\begin{table*}[t!]
\centering
\caption{Authorization-to-effect paths selected for distinct mechanisms. Each
contract records its detailed scope and evidence obligations. RBAC denotes role-based
access control; PubAck denotes a downstream publish acknowledgment.}
\label{tab:grounded-routes}
\footnotesize
\setlength{\tabcolsep}{2.6pt}
\begin{tabularx}{\textwidth}{@{}
  >{\raggedright\arraybackslash}p{0.12\textwidth}
  >{\raggedright\arraybackslash}p{0.225\textwidth}
  >{\raggedright\arraybackslash}p{0.20\textwidth}
  >{\raggedright\arraybackslash}p{0.225\textwidth}
  >{\raggedright\arraybackslash}X@{}}
\toprule
Path & Authorization decision & Carrier & Effect frontier & Open-path failure \\
\midrule
GitHub MCP & Application approves head $h_1$ & Late-bound merge request & Atomic compare-and-merge~\cite{githubmcp,githubmerge} & Atomic Bind absent \\
NATS redelivery & Application authorizes request $r$ for one logical effect & Redelivered work carrying $a$ & PubAck downstream append~\cite{natsdelivery,natspublishing,natssource} & Logical effect unbound \\
Kubernetes mutation & Admission authorizes the checked object & Request through mutation chain & Final validation before persistence~\cite{k8sadmission} & No final-object recheck \\
Kubernetes RBAC & RBAC authorizer allows create & Already-admitted request & ConfigMap persistence~\cite{k8sadmission,k8srbac} & Old work not rechecked \\
NATS purge & Application authorizes one dispatched delivery effect & Callback carrying $a$ & PubAck downstream append~\cite{natssource} & Local state omits callback \\
Kafka produce & WRITE ACL authorizer allows request & In-flight \texttt{ProduceRequest} & Append to target leader's local log~\cite{kafkakip801,kafkadeleteacls} & Model stops before append \\
\bottomrule
\end{tabularx}
\end{table*}
\subsection{From Certificates to Execution}

\noindent\textbf{Scoped claim records.}
When a success response reports closure, its claim record names the kind,
subject, analysis contract, policy, frontier, scope $\kappa$, and configuration
version $\beta$. Its claim and context fields must match the required template. After a
switch to $\beta'$, the deployment may rely on that
claim only with an explicit preservation proof. The original claim remains tied
to the configuration version $\beta$ under which it was emitted. A checked winning
certificate establishes only the claims emitted by its strategy and only within
the validated contract.

\noindent\textbf{Per-trace conformance.}
Deployment traces record identities, event order, frontier placement, and
effects. The trace checker rejects mismatches with the repaired contract and
deployment requirements. An accepted trace establishes conformance only for
that execution. All provider executions require a separate proof or complete
online checking.

\noindent\textbf{Worked example.}
For the GitHub MCP-to-REST path, $g$ is the grant lineage for reviewed head
$h_1$, $a$ its pull-request~42 authorization instance, $c$ the dispatched request,
and $F$ the atomic merge. The policy
accepts $\mathsf{Merge}(42,h)$ only for $h=h_1$~\cite{githubmcp,githubmerge}.
After both worlds return $h_1$, the same information state contains A, where the
head remains $h_1$, and U, where it may later move to $h_2$. The interface
exposes read, merge without a head-SHA precondition, and Abort:
\begin{certificatebox}{Losing certificate: GitHub MCP-to-REST path}
\textbf{Obligations}
\begin{itemize}[leftmargin=1.25em,itemsep=0.25ex,parsep=0pt,topsep=0.35ex]
  \item \textbf{Required:} merge the authorized head $h_1$.
  \item \textbf{Safety:} never merge the changed head $h_2$.
\end{itemize}
\textbf{Same observation.} Both worlds return $h_1$ at the final read.
\par\noindent
\textbf{Hidden divergence.} U may still move to $h_2$ afterward.
\par\smallskip\noindent
\textbf{Every exposed action loses}
\begin{itemize}[leftmargin=1.25em,itemsep=0.25ex,parsep=0pt,topsep=0.35ex]
  \item \texttt{merge} may commit $h_2$ (\emph{unsafe});
  \item \texttt{read} leaves the dilemma unresolved (\emph{cannot guarantee completion});
  \item \textsf{Abort} blocks the required merge (\emph{incomplete}).
\end{itemize}
\par\smallskip\noindent
\textbf{Missing power.} Atomic binding to authorized SHA $h_1$.
\end{certificatebox}

The Bind candidate adds \texttt{sha=$h_1$}, merging $h_1$ in the required
world and rejecting a moved head before $F$, without changing the public
response or effect semantics. The SHA becomes an atomic merge precondition rather than
another observed value. Native REST traces link the authorized head through this check
to the merged head. A separate run of the pinned GitHub MCP server confirms
the late-bound behavior end to end. This contract promises a safe merge. It makes no
closure claim.
The evaluation applies the analysis to six paths and studies Kafka in depth.

\section{Evaluation}
\label{sec:results}

\begin{table*}[t]
\centering
\caption{Open paths, boundary changes, and outcomes. ``Base'' and ``Repaired''
count forbidden effects after the native and repaired success signals.
``Complete'' counts designated authorized trials that finish. The Kafka row
previews Study~B. GitHub counts are REST replay pairs. One run through the
official GitHub MCP server provides an additional end-to-end confirmation.}
\label{tab:deployed-paths}
\footnotesize
\setlength{\tabcolsep}{3.2pt}
\begin{tabular}{lllllrr}
\toprule
Path & Open continuation & Base & Policy goal & Repair &
Repaired & Complete \\
\midrule
GitHub MCP-to-REST & Late-bound pull-request target & 25/25 moved head &
Merge authorized $h_1$ & Bind SHA & 0/25 & 25/25 \\
NATS redelivery & Redelivery carrying $a$ & 10/10 duplicate &
One output per $a$ & Bind ID & 0/10 & 10/10 \\
Kubernetes mutation & Post-check mutation chain & 25/25 privileged &
Persist safe object & Re-enter & 0/25 & 25/25 \\
Kubernetes RBAC & Paused admitted request & 30/30 persisted &
No persistence under old authority & Re-enter & 0/30 & 30/30 \\
NATS purge & Already-dispatched callback & 30/30 appended &
{\scriptsize$\mathsf{InstanceEffectClosed}(a)@(\kappa,\beta)$} & Tracked Drain & 0/30 & 30/30 \\
Kafka produce & Old authorized request & 10/10 appended &
{\scriptsize$\mathsf{LineageClosed}(g)@(\kappa,\beta)$} & Effect Fence & 0/10 & 10/10 \\
\bottomrule
\end{tabular}
\end{table*}

\subsection{Evaluation Design}
\label{sec:evaluation-design}

\noindent\textbf{Questions and evidence sets.}
The evaluation has two primary studies and one secondary assurance check.
Study~A asks whether the three placements eliminate the challenged forbidden
effect in six structurally different authorization-to-effect paths without
blocking designated work.
Study~B examines Kafka in depth: whether Effect Fence closes paths from
previously authorized requests, preserves unrelated service, and at what cost
in the test deployment. The secondary check covers abstention when evidence is
missing, derivation of the declared closure result for the NATS purge base
contract and three variants, and rejection of corrupted translations and
certificates. These evidence sets answer different questions and are reported separately.

\noindent\textbf{Path selection and contract scope.}
Paths were chosen to cover different mechanisms. The six live paths span late
binding in the GitHub MCP-to-REST path;
post-check mutation and already-admitted execution in Kubernetes; redelivery
and callback-held work in NATS; and distributed in-flight execution in Kafka.
Together they exercise Bind, Re-enter, and Fence across four systems.

A preregistered cross-layer audit follows one policy-relevant control through 20
operations, from the native API through a software development kit
(SDK) to a tool. Native-to-SDK representations retain it in 20/20 cases; tool
forms preserve 11/20, omit two, weaken two, and move five to a different decision
point (see Appendix~\ref{app:confidence-detail}).

Each path fixes a provider boundary, effect frontier, policy, workload, bounds,
exposed controls, candidate repairs, and exclusions before analysis. Every
policy-relevant provider premise used in a grounded verdict is supported by at
least two of three evidence categories: documentation, pinned schema or source,
and controlled execution. The evidence records list available capabilities,
premises, and exclusions. Missing support for a
transition, ordering fact, effect, authorization link, return meaning, or
exclusion yields \textsc{Unsupported} before finite analysis.
Table~\ref{tab:grounded-routes} records each authorization decision, carrier,
frontier, and open-path failure. Public closure claims include their scope
$\kappa$ and configuration $\beta$ in Table~\ref{tab:deployed-paths} and the
machine-readable records.

\noindent\textbf{Common challenge protocol.}
Study~A constructs required and forbidden executions that remain
indistinguishable at the selected boundary. Each challenge follows the same
sequence: stall work after it has been authorized, change authorization or
provider state, establish the visible success condition, release the old work,
and observe whether it crosses the effect frontier. Base and Repaired count
forbidden effects after native and repaired success; and Complete counts designated
authorized completions. These challenges can show that a model stops before the
frontier, as Kafka did, but cannot establish that no additional provider path exists.

\noindent\textbf{Kafka campaign.}
Study~B compares Applied-State Fence (waiting for fixed-set authorizer
convergence) with Effect Fence under matched load, stalled requests, placement
races, crashes, and a fixed controller/broker set. Component removals test
blocking new use, recording the authorization version, and draining earlier requests.
A bounded extension tests one broker rejoin after the native authorizer becomes
ready. Five configurations measure protected-path cost, and four
matched pairs compare the policy-scoped gate with a lab-only leader-wide gate
that also blocks unrelated work. We preregistered the high-contention sample
size before seeing results. The claim covers synchronous nontransactional
append to the target leader's local log for one principal and topic, the
complete matching ACL set, and the fixed controller and broker set.

\noindent\textbf{Grounding and validation plan.}
A fixed abstention audit selects eight operations, one from each stratum of a
16-operation frame. It tests whether missing authorization links or frontiers yield
\textsc{Unsupported} rather than being silently invented, rather than estimate
applicability.
Separate finite suites test classification accuracy on 24 cases, all 17 game
certificates, 42 certificate corruptions, seven semantic mutations, and repair
preservation and minimality (Appendix~\ref{app:confidence-detail} gives the
audit's selection rule and denominators and summarizes the 12-contract suite;
the artifact records its hashes).

\subsection{Study A: Open Paths Across Providers}

\noindent\textbf{Question and coverage.}
Study~A asks whether different provider mechanisms leave the same kind of path
from authorization to a forbidden effect, and whether a control at the last
policy-relevant branch can close it without blocking required work.
Table~\ref{tab:deployed-paths} covers six paths. Bind fixes identity before
late resolution or repetition; Re-enter checks after relevant state changes and
Fence restricts or drains work that has passed the authorization check.
The placements answer different questions and are not interchangeable: Bind
preserves what was authorized, Re-enter recomputes permission after the facts
change, and Fence accounts for work that will not pass through authorization
again. Each challenge pairs an open continuation with required authorized
work, so success requires more than blocking the path.

\noindent\textbf{Bind the authorized identity.}
Bind applies when the provider resolves the target late or may execute the same
authorized instance more than once. GitHub's native merge endpoint accepts an
atomic head-\texttt{sha} precondition, whereas the pinned MCP tool omits
it~\cite{githubmcp,githubmerge}. Table~\ref{tab:deployed-paths} reports 25 paired
REST replays testing the native precondition: binding to $h_1$ preserves every
unchanged-head merge and rejects every moved-head merge. A separate end-to-end
run of the pinned official GitHub MCP server tests the MCP-to-REST path and
confirms late binding. The server reads $h_1$. After a move to $h_2$, its merge
commits $h_2$. Unchanged-head and \texttt{create\_issue} control trials pass. Native REST with stale
\texttt{sha=$h_1$} returns 409 without merging.
Both paths use the same merge semantics and workload. Only the native boundary
exposes the atomic precondition, so the contrast isolates control lost at the
abstraction boundary.

NATS redelivery presents the same placement problem through repetition rather
than late target resolution. Here $a$ is the application's prior authorization
of request $r$ and payload $p$ to cause exactly one logical effect $e(r,p)$. Each
delivery attempt carries the same $a$. In 10/10 schedules the base produces two
durable effects. An idempotency ID bound to $(a,r,p)$ admits both attempts but
commits one effect within the configured duplicate window. After expiry both
commit~\cite{natsdelivery,natspublishing,natssource} (Appendix~\ref{app:confidence-detail}).
Including $a$ and $p$ binds deduplication to the authorized logical effect,
rather than treating all equal-looking payloads as the same authority. The
expiry run is the negative control for the declared time bound.

\noindent\textbf{Re-enter after state changes.}
Re-enter applies when policy-relevant state can change after the original check.
Kubernetes mutating admission precedes final validating
admission~\cite{k8sadmission}. Checking the final object rejects all 25
privileged mutations while preserving all 25 required objects.
This path changes the object after its first check. The RBAC path below instead
changes the authority state while an already-admitted request survives. Both
need Re-enter, but at different last policy-relevant branches.

The RBAC path moves the same problem into already-admitted execution. In
Kubernetes~v1.36.1, a validating-admission hook pauses a ConfigMap create after
RBAC permits it~\cite{k8sadmission,k8srbac}. We delete the sole RoleBinding and a fresh
create is denied, but releasing the earlier request persists the object in
30/30 trials. Rechecking RBAC in the callback gives 0/30 forbidden effects and
30/30 authorized completions. The result covers one principal, RoleBinding, object,
and API-server path, excluding regrant and alternate authorization. The recheck
is not atomic with storage.

\noindent\textbf{Fence work already past authorization.}
Fence addresses authorized work that has crossed the decision point. NATS
purge completes the stream-storage operation but does not recall a delivery
already dispatched into application code. After a subject-filtered purge, the
trace records zero stream messages, \texttt{num\_ack\_pending}, and
\texttt{num\_pending} for the selected consumer~\cite{natssource}. Releasing the
held callback nevertheless causes the application to publish the controlled
downstream record in all 30 trials. The application authorized one
downstream record, and the callback still carries that authority after the
stream state changes. The zero counters describe only NATS state and do not show
that the callback has finished. If the model did not link the callback to that
authorization, the callback would cease to be a carrier and instance-effect
closure would appear true.
Bind and Tracked Drain both satisfy the finite contract. Table~\ref{tab:deployed-paths}
reports Tracked Drain because only it was tested on live traces and directly
supports the post-purge closure claim. Waiting for the tracked callback to
finish yields 0/30 post-repair forbidden effects, completes all 30 required
trials, and leaves unrelated work unblocked. The base contract is
\textsc{Unrealizable} for this stronger closure requirement. Native purge need
not cancel an already-dispatched callback~\cite{natssource}.

\begin{rqanswer}{Study A}
\textbf{Finding.} Across six paths, work carrying earlier authority can
still cross the effect frontier after the authorization decision or a later
local completion signal.
\par\noindent
\textbf{Intervention.} Bind, Re-enter, and Fence close different last
policy-relevant branches. For every studied path, the matching placement
eliminates the challenged forbidden path while designated authorized work
completes.
\par\noindent
\textbf{Scope.} This establishes structural coverage and feasibility
within each path's declared scope.
\end{rqanswer}

\subsection{Study B: Kafka Effect Fence}
\label{sec:study-b}

\noindent\textbf{Question and scope.}
Study~B asks how far Kafka lineage closure must extend. We compare native deletion,
fixed-set authorizer convergence, and Effect Fence. We then test component
removals, along with concurrency, crashes, a bounded rejoin, admission scope,
and cost.
The claim covers synchronous nontransactional append to the target leader's
local log for one principal and topic, the complete matching ACL set, and a
fixed controller and broker set.

\begin{figure*}[t]
\centering
\begin{tikzpicture}[
  font=\sffamily\footnotesize,
  >={Latex[length=1.6mm]},
  flow/.style={draw=black!55, rounded corners=1.5pt, align=center,
    minimum height=9.5mm, text width=2.85cm, inner sep=3pt},
  normal/.style={flow, fill=black!2},
  visible/.style={flow, fill=black!5, draw=black!68},
  unsafe/.style={flow, fill=black!9, draw=black!82, line width=0.7pt},
  safe/.style={flow, fill=white, draw=black!82, line width=0.8pt},
  flowarrow/.style={-{Latex[length=1.6mm]}, semithick, draw=black!65},
  lane/.style={font=\sffamily\bfseries\footnotesize, anchor=east}
]
  \node[lane] (base-label) at (-0.15,2.05) {Base};
  \node[normal]  (b1) at (1.55,2.05) {\texttt{DeleteAcls}\\completes};
  \node[visible] (b2) at (5.25,2.05) {API returns\\success};
  \node[unsafe]  (b3) at (8.95,2.05) {Stale broker accepts\\revoked request};
  \node[unsafe]  (b4) at (12.65,2.05) {\textbf{FORBIDDEN EFFECT}\\Revoked append\\10/10};
  \draw[flowarrow] (b1) -- (b2);
  \draw[flowarrow] (b2) -- (b3);
  \draw[flowarrow] (b3) -- (b4);

  \node[lane, align=right] (state-label) at (-0.15,0) {Applied-State\\Fence};
  \node[normal]  (s1) at (1.55,0) {Old request authorized\\and paused};
  \node[normal]  (s2) at (5.25,0) {Every fixed-set authorizer\\applies the deletion};
  \node[visible] (s3) at (8.95,0) {\textbf{typed claim}\\$\mathsf{FutureUseClosed}(g)$\\$@(\kappa,\beta)$};
  \node[unsafe]  (s4) at (12.65,0) {\textbf{FORBIDDEN EFFECT}\\Release old request\\append to target leader's local log};
  \draw[flowarrow] (s1) -- (s2);
  \draw[flowarrow] (s2) -- (s3);
  \draw[flowarrow] (s3) -- (s4);

  \node[lane, align=right] (effect-label) at (-0.15,-2.05) {Effect\\Fence};
  \node[normal]  (e1) at (1.55,-2.05) {Block new use of\\revoked authority};
  \node[normal]  (e2) at (5.25,-2.05) {Apply deletion and wait\\for earlier requests};
  \node[visible] (e3) at (8.95,-2.05) {\textbf{typed claim}\\$\mathsf{LineageClosed}(g)$\\$@(\kappa,\beta)$};
  \node[safe]    (e4) at (12.65,-2.05) {\textbf{CLOSED}\\0/10 post-repair\\forbidden effects};
  \draw[flowarrow] (e1) -- (e2);
  \draw[flowarrow] (e2) -- (e3);
  \draw[flowarrow] (e3) -- (e4);

  \begin{scope}[on background layer]
    \node[fit=(b1)(b4), fill=black!1, draw=black!28, densely dashed, rounded corners=2pt,
      inner xsep=5pt, inner ysep=6pt] {};
    \node[fit=(s1)(s4), fill=black!1, draw=black!32, densely dashed, rounded corners=2pt,
      inner xsep=5pt, inner ysep=6pt] {};
    \node[fit=(e1)(e4), fill=black!1, draw=black!55, rounded corners=2pt,
      inner xsep=5pt, inner ysep=6pt] {};
  \end{scope}
\end{tikzpicture}
\caption{Kafka's three return conditions reach progressively deeper boundaries. Applied-state
convergence establishes only $\mathsf{FutureUseClosed}(g)@(\kappa,\beta)$, so a paused
authorized request can still append. Effect Fence blocks revoked use and waits
for earlier requests before establishing
$\mathsf{LineageClosed}(g)@(\kappa,\beta)$.}
\Description{Three horizontal Kafka lanes compare Base, Applied-State Fence,
and Effect Fence. Base returns success while a stale broker accepts a revoked
request, producing 10 of 10 forbidden appends. Applied-State Fence establishes
future-use closure, but a paused old request still appends. Effect Fence blocks
new use of revoked authority and waits for earlier requests before establishing lineage
closure, producing zero of ten post-repair forbidden effects.}
\label{fig:kafka}
\end{figure*}
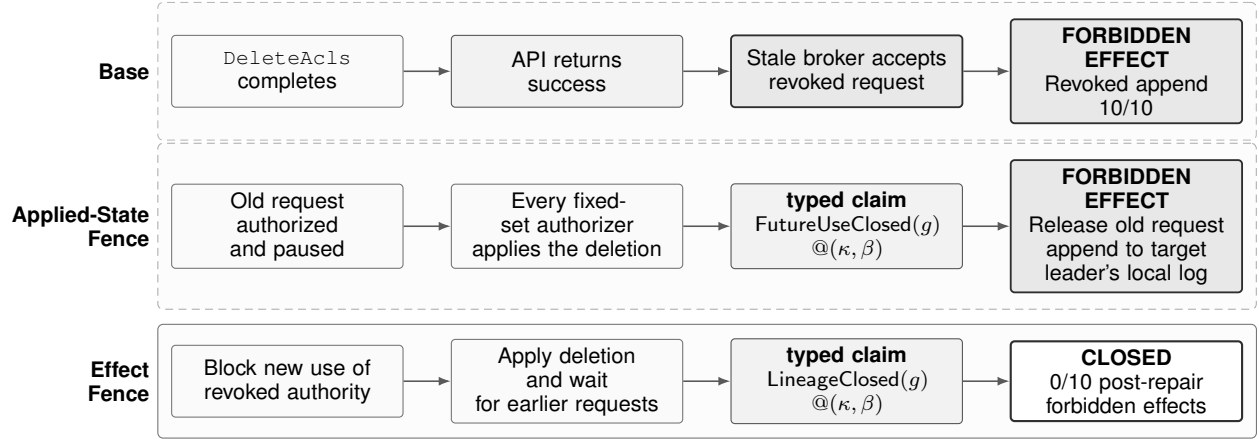

\noindent\textbf{Closure depth and frontier adequacy.}
Within scope $\kappa$, Kafka authorizers apply ACL records from the metadata log
and may lag the controller~\cite{kafkakip801}. Native
\texttt{DeleteAclsResult.all()} reports
completion of the requested deletions. Convergence requires checking every
broker~\cite{kafkadeleteacls}. In the base path, a lagging authorizer admits
the request and produces the forbidden append in 10/10 trials.

Applied-State Fence waits for every authorizer in the fixed broker set to apply
the deletion. For lineage $g$ covering the complete matching ACL set, this establishes
$\mathsf{FutureUseClosed}(g)@(\kappa,\beta)$, but the application requires
$\mathsf{LineageClosed}(g)@(\kappa,\beta)$: no request previously authorized by
$g$ may append after the fence returns, while unrelated service remains usable.
On Kafka~4.3.1, we pause a \texttt{ProduceRequest} after authorization, wait for
Applied-State Fence success, then release it. The old request still reaches the
target leader's local log (Figure~\ref{fig:kafka}):
\begin{equation*}
  \mathsf{FutureUseClosed}(g,q)\ \not\Rightarrow\
  \mathsf{LineageEffectClosed}_{\varphi_{\rm app}}(g,q).
\end{equation*}
Topic authorization precedes \texttt{handleProduceAppend}. The Applied-State
analysis correctly established authorizer convergence, but its contract stopped
before the paused request's remaining path to append: its proposed model boundary
$B_{\rm applied}$ preceded the effect frontier $F_{\rm append}$.
Adding that path changes the verdict: future-use closure is insufficient.
Thus the three
returns establish, respectively, deletion completion, future-use closure, and
lineage closure at append; alternate credentials, wildcards, and future
membership remain outside the claim.

\noindent\textbf{Effect Fence construction.}
The omitted path from authorization to append makes Applied-State Fence
insufficient. Effect Fence records each request's authorization epoch---the
authority version at its check---and holds a lease through append. Each lease
identifies one admitted request and remains live through append. Revocation
closes the epoch, rejects new leases under it, and waits for existing leases to
drain. The gate closes future use and the leases account for work already admitted.
Fence succeeds
only after every fixed-set broker shows that it applied the deletion
in its current boot and no closed-epoch lease remains. Together, those conditions
rule out new use and effects from earlier work, establishing
$\mathsf{LineageClosed}(g)@(\kappa,\beta)$. The gate and leases enforce the
ordering. Broker evidence supports the closure claim. \system checks the
condition, and trace checking links the contract to the Java execution
(Appendix~\ref{app:confidence-detail}).

\begin{table}[!t]
\centering
\caption{Old-instance path ablation at success (8 deterministic runs per
variant).}
\label{tab:fence-ablation}
\footnotesize
\setlength{\tabcolsep}{3pt}
\begin{tabular}{@{}lcc@{}}
\toprule
Variant & Old path closed & Unrelated admitted \\
\midrule
Admission closed, no drain & \xmark & \cmark \\
Epoch gate, no drain & \xmark & \cmark \\
Drain once, admission open & \xmark & \cmark \\
Lab-only leader-wide drain & \cmark & \xmark \\
Policy-scoped lease gate & \cmark & \cmark \\
\bottomrule
\end{tabular}
\end{table}

\noindent\textbf{Safety and composition evidence.}
Table~\ref{tab:fence-ablation} separates closing admission, epoch checking,
and draining old instances. Let $a_0$ be the old authorization instance and $a_1$
the authorization instance carried by a required work item outside $a_0$'s
closed epoch. Instance $a_1$ completes in 8/8 for every variant. Every partial
defense leaves an open carrier of $a_0$ in 8/8. Both full defenses make $a_0$
effect-closed. Only the leader-wide drain blocks unrelated work.
Full lineage closure also requires fixed-set evidence for
$\mathsf{FutureUseClosed}(g)@(\kappa,\beta)$. Omitting either requirement leaves
an admitted execution that violates the claim. ``Leader-wide'' names only this
test variant.

The concurrency probe pauses after authorization, once the request has captured
its epoch and while it still holds the gate. This ensures that revocation cannot
change the captured epoch.
Across 120 stalled-lease trials, Fence remains pending while an old lease is
live. Unrelated work completes, and releasing the lease permits success. The
placement test observes append to the target leader's local log before
acknowledgments. A pre-append crash creates no effect. A post-append crash
preserves it. Restart invalidates evidence from the prior boot until a fresh
authorizer snapshot loads, while service remains
available. Thus the measured effect is synchronous nontransactional append to
the target leader's local log rather than acknowledgment, high-watermark visibility, or transaction commit.

The executable Kafka protocol checker covers 25 targeted cases and 1,000 generated
schedules and 10 live control trials confirm that required service remains
available. Within the finite contract, Lean proves atomic capture, rejects
post-closure leases, and permits success only after every broker applies the
deletion and no old lease remains. Together, fixed-set current-boot reports and
an empty old-lease set exclude old-epoch append.
In 10 live trials per rejoin condition, base rejoin produces 10/10 unsafe
effects. Fence with or without a broker-join gate produces 0/10 unsafe effects
and 10/10 authorized completions. The coordinator records the matching ACL set,
waits for concurrent deletions, and invalidates return if any ACL is re-added.
Kafka readiness provides the fresh authorizer snapshot the join gate requires.
The gate adds no observed protection in these trials. Other membership changes
require a separate admission proof.

\noindent\textbf{Trace linkage and service scope.}
Per-trace checking connects authorization, same-thread epoch capture, lease,
deletion, success, and append. The append path reacquires the gate, accepts only
the captured current epoch, and holds its lease through synchronous append to
the target leader's local log. Revocation must wait for the lease or advance the
epoch first. The trace checker
rejects a mismatched identity, epoch, boot, order, frontier, or effect.
This result covers only the recorded execution.

The scope assumes one authorizer per Java Virtual Machine (JVM) boot, snapshot
readiness before traffic, and a fixed controller and broker set. It excludes authorizer
replacement, cross-thread handoff, transactions, acknowledgments,
high-watermark visibility, membership reconfiguration, rollback, and unbounded
execution. Fairness affects only progress. Counters and logs serve only as evidence.

In four clean-boot pairs with alternating run order, we hold a target lease for at least
1.5\,s while sending unrelated work. That work finishes before release in 4/4
policy-scoped runs (median 469.8\,ms) but 0/4 lab-only leader-wide runs (median
1982.0\,ms after release), a median within-pair difference of $+1517.4$\,ms.
Fence remains pending. After release, the final trace contains both records.
Thus the scoped gate preserves unrelated work during the 1.5-s hold.

\begin{table}[t]
\centering
\caption{Cost on the live protected path ($n=12$ matched runs per variant and
concurrency level). Entries are median paired deltas against unmodified
Kafka (Stock); $c$ is producer concurrency and p99 is 99th-percentile latency.}
\label{tab:kafka-incremental}
\footnotesize
\setlength{\tabcolsep}{1.0pt}
\renewcommand{\arraystretch}{1.03}
\begin{tabularx}{\columnwidth}{@{}>{\raggedright\arraybackslash}Xrrrr@{}}
\toprule
Variant & \multicolumn{2}{c}{$\Delta$ throughput} & \multicolumn{2}{c}{$\Delta$ p99 (ms)} \\
\cmidrule(lr){2-3}\cmidrule(l){4-5}
 & $c=32$ & $c=128$ & $c=32$ & $c=128$ \\
\midrule
Instrumentation & -3.4\% & +2.3\% & -1.0 & -8.7 \\
Epoch capture & -2.5\% & +4.7\% & +1.2 & -10.6 \\
Old-instance fence & -4.2\% & -11.8\% & +10.4 & +36.5 \\
Full fixed-set configuration & -18.9\% & -15.0\% & +89.1 & +2.3 \\
\bottomrule
\end{tabularx}

\end{table}

\noindent\textbf{Cost, timing, and deployment.}
Only the old-instance fence and full fixed-set variants provide security
guarantees. The first drains old work through the epoch gate, while the second
adds current-boot authorizer evidence for lineage closure.
Table~\ref{tab:kafka-incremental} reports paired
medians from 12 matched runs per variant and contention level. The pinned
Kafka~4.3.1 deployment has one Kafka Raft (KRaft) controller, three brokers,
and one three-replica partition (minimum in-sync replicas: two). Measurements
include connection and metadata startup. Host resources and scheduling are
uncontrolled. At $c=32,128$, respectively, the old-instance fence changes
throughput by $-4.2\%,-11.8\%$ and p99 by $+10.4,+36.5$\,ms; the full variant changes
them by $-18.9\%,-15.0\%$ and $+89.1,+2.3$\,ms. Component differences change
sign and span broad ranges, so no individual cost can be isolated. Batching or
shared-memory lease tracking may reduce overhead but remain untested (see
Appendix~\ref{app:confidence-detail} for details about absolute medians and ranges).

Once the fixed-set evidence is current, the median interval from \texttt{DeleteAcls}
completion to Fence success is 18.6\,ms ($n=10$, range 7.7--49.6\,ms). Median
time to Fence success is 337.9\,ms after metadata recovery and 5.26\,s
after a fresh-authorizer boot, and those intervals start at different events.

The 644-line prototype uses a custom authorizer and a Java agent at Kafka's
produce-append hook, leaving Kafka unmodified. Component removals, faults, broker
rejoin, and held requests test the controls, invariant, and composition within
the stated scope. Trace checking links authorization to append. This shows
feasibility for the declared path. Production requires a maintained hook and a
protocol that keeps membership evidence aligned with the claim.

\begin{rqanswer}{Study B}
\textbf{Finding.} Fixed-set authorizer convergence closes future use, but a
request admitted under the old epoch can still append.
\par\noindent
\textbf{Intervention.} Effect Fence combines current-boot deletion
evidence, an epoch gate, and an old-lease drain. It establishes lineage closure
for the declared path while designated work completes. Its scoped gate admits
unrelated work that the lab-only leader-wide gate blocks.
\par\noindent
\textbf{Scope.} The claim covers synchronous nontransactional append to the
target leader's local log under the declared fixed controller and broker set.
\end{rqanswer}

\subsection{Secondary Check: Fail-Closed Validation}

\noindent\textbf{Does missing evidence stop analysis?}
Seven of eight operations in the 16-operation frame lack evidence
for the proposed effect path, the link from work to its authorization, or the
return meaning and therefore yield \textsc{Unsupported}. The audit checks that
missing evidence stops analysis. It does not measure applicability or verdict frequency.
Only NATS purge supports
the proposed analysis: its base and native-terminal variants are
\textsc{Unrealizable}, and Bind and Tracked Drain are \textsc{Realizable}. Only
Tracked Drain is validated against live traces. Separate GitHub issue-creation
and Kubernetes scaling checks serve as positive controls.

\noindent\textbf{Can malformed inputs or biased repair sets produce a false positive?}
Across 24 finite cases, \system matches exhaustive search, including one losing case
that requires three worlds and repairs that combine controls. Lean validates all
17 translated game certificates (2 winning, 15 losing). Fourteen separate
grounding diagnostics record missing effect histories. Brute force matches the solver's
Pareto-minimal sets for three repair cases, and all 42 corrupted certificates are
rejected. The omitted Kafka path invalidates the Applied-State contract before solving.
None of 38 subsets drawn from 16 irrelevant controls wins; after removing the
only repair that closes the path, none of the 70 tested subsets wins. A corrupted exclusion also
produces no false positive.

\noindent\textbf{Are typed claims derived?}
Pre-search validation rejects five public-claim mutations and two attempts to
relabel or erase effects. All 21 evaluated repairs preserve effects. The four NATS
purge contracts---base, native-terminal, Bind, and Tracked Drain---declare
$\mathsf{InstanceEffectClosed}(a_{\rm purged\_delivery})$.
From the contract's issuance, carrier, suffix, permission, policy, and coverage
relations, Lean derives the closure judgment and checks $\mathsf{ClaimSafe}$.
Bind and Tracked Drain satisfy the requirement; the base and native-terminal
variants do not. Kafka is checked
separately through the Effect Fence invariant and fixed-set composition
(Appendix~\ref{app:finite-detail}) and bounded trace evidence
(Appendix~\ref{app:confidence-detail}).

\section{Discussion}

\noindent\textbf{A closure response must say what has closed.}
  Lineage closure combines future-use closure with effect closure for every
  in-scope instance issued from that lineage. Instance-effect closure covers one.  Success alone establishes no closure. Claims name subject, kind, policy,
  frontier, scope, and configuration and persist after return. Configuration
  changes must preserve them. Closure does not require quiescence. Bind, Re-enter, and Fence close such paths by binding, rechecking, or waiting while required  work continues. Kafka's scoped gate preserves unrelated work.

\noindent\textbf{Verification must reach the effect.}
In our cross-layer audit, tool forms omitted, weakened, or moved the selected
policy-relevant control in 9 of 20 operations.
Mediating every exposed call is insufficient
when identity resolves later, provider state changes after the check, or admitted
work continues after the call. A machine-checked proof can be correct for its
contract even when that contract stops before the application effect. Frontier selection is part of the
security argument. End-to-end mediation must preserve
the checked authorization condition to the effect frontier, recheck after the last policy-relevant
change, or delay return until admitted work can no longer produce a forbidden
effect. Effect closure also covers work already admitted. Composing claims
requires related subjects and kinds, aligned policy,
frontier, scope, and configuration, and proof of preservation.

\noindent\textbf{Boundary realizability identifies missing power.}
Within the finite contract, a losing certificate rules out every strategy
available through the modeled interface. An authorization twin pairs
indistinguishable required and unsafe worlds, so every available choice permits
a forbidden prefix or blocks required work. The tested GitHub MCP interface
exposes merge without the native SHA precondition, so a mediator limited to it
cannot recreate atomic binding. Without binding, recheck, or fence, safety and
completion are incompatible at that boundary. The interface must expose more
control or return a weaker closure claim. Repairs may add state, epochs, or
leases but must preserve workload, responses, claims, and the meaning and
representation of effects. A changed public claim defines a new contract. A
mediator change defines a new configuration version. A deployment may rely on
earlier guarantees only with a preservation proof. Least privilege limits what
may be authorized. Boundary realizability asks whether the interface preserves
that limit through the effect frontier.

\noindent\textbf{Provider-local completion does not imply application-level closure.}
Provider completion may be correct while admitted work can still reach
  a rejected effect. Closure must be
  established separately from provider conformance. We privately reported our findings to maintainers (Appendix~\ref{app:ethics}).

\noindent\textbf{Threats to validity.}
\emph{Construct validity} depends on modeled paths, authorization links, the
frontier, and post-return behavior. \textsc{Unsupported} flags missing evidence
for represented premises. It cannot detect omitted paths, so provider-wide
completeness remains unproved. \emph{Internal validity} depends on accurate
contracts. Machine checking derives four NATS closure results and validates 17
certificates; traces cover observed runs only. One analyst authored
contracts and challenges. Preregistration and mutation tests constrain but do
not remove bias. \emph{External validity}: Kafka tests crashes and a
bounded fixed-set rejoin. Network partitions, replication faults, and fixed-set
changes are outside scope. Prevalence is not estimated.
\emph{Conclusion validity}: performance is descriptive. Repair minimality
depends on the declared set and cost order.

\section{Related Work}
\label{sec:related}

Work on \emph{authority lifetime} models revocation and ongoing use.
\emph{Enforceability over a supplied interface} assumes that interface is monitored; and
\emph{supervisory control} assumes a DES and specification. \system instead asks
what must close, whether the model reaches the effect frontier, and whether boundary
controls can realize the claim while required work completes. These questions must
be answered before finite control begins: a correct supervisor over the wrong
property or frontier proves the wrong application guarantee.

\noindent\textbf{Authority lifetime and open paths.}
Revocation work tracks issued or migrated authority
\cite{gligor1979revocation,flask1999}. Usage control (UCON) supports continuing
decisions during use~\cite{uconabc2004,distributedusagecontrol2006}. Under a
post-revocation policy, complete revocation matches lineage closure when it
stops new issuance and blocks every later rejected effect from issued authority.
\system instead asks whether the provider interface can make the
requested closure claim true. TOCTOU studies check/use changes
\cite{bishopdilger1996,mindthegap2025}. Work on the authorization--execution gap
examines broader divergence~\cite{authorizationexecutiongap2026}. GitHub MCP and
Kubernetes exhibit such changes~\cite{githubmcp,githubmerge,k8sadmission},
whereas NATS and Kafka retain issued authority. Effect closure gives these
authorization-bearing paths a common completion condition.

\noindent\textbf{Enforceability over a supplied interface.}
Complete mediation, runtime enforcement, and usage control tie guarantees to
observable or controllable events
\cite{schneider2000,ligatti2005,saltzerschroeder1975,basin2013enforceable}.
Linearizability and interposition study histories, missing hooks, and
races~\cite{herlihywing1990,garfinkel2003,dam2009monitor}. Agent systems enforce
rules, monitors, or permission graphs over chosen interfaces or
contracts~\cite{agentspec2026,forge2026,fava2026}. These
approaches enforce policies over a supplied interface. \system first asks whether
that interface reaches the required effect frontier and exposes enough
observation and control to report closure without blocking required work.

\noindent\textbf{Supervisory control and security.}
Ramadge--Wonham control and Lin--Wonham observability provide our
basis~\cite{ramadgewonham1987,linwonham1988}. A survey of 208
papers covers DES security from opacity to attacks and recovery
\cite{oliveira2023dessecurity}. Supervisory control also supports dynamic
network defense~\cite{rasouli2014dynamic}. We use this established theory to
decide boundary realizability. Before applying it, we construct the
provider-to-effect contract from provider evidence. It states the requested
closure and required completion, grounds the frontier, and poses the
realizability question.

\noindent\textbf{Mechanisms and abstraction.}
Binding, revalidation, and fencing are established mechanisms
\cite{schneider2000,ligatti2005,birrellnelson1984rpc}. Bind, Re-enter, and Fence
locate where they act along the path. Commit-time authorization~\cite{cta2026}
can realize Re-enter when an exposed commit gate rechecks authority for every
relevant effect. PORTICO~\cite{lingeringauthority2026},
AID-Guard~\cite{aidguard2026}, and \textsc{SoundGate}~\cite{stopmeansstop2026}
assume supplied bindings, tools, contracts, or framework controls.
\system instead asks whether the boundary reaches the effect frontier and exposes
such a gate or another repair that preserves safety and required completion.
Isolation and access control
constrain large language model (LLM) applications and MCP servers
\cite{isolategpt2025,agentbound2026}. Other MCP work catalogs threats
\cite{mcplandscape2026}.
End-to-end and tool-contract work asks whether interfaces preserve policy-relevant
choices~\cite{saltzerreedclark1984,howellkotz2000,dci2026}. End-to-end
authorization asks whether authority reaches an operation intact. Effect closure
asks when authority already issued on that path has no remaining route to a
policy-rejected effect. We test whether the modeled boundary remains realizable
within its stated scope. Coordination-avoidance work instead asks which
guarantees require coordination~\cite{bailis2014}.

\section{Conclusion}

Authorization state can converge while admitted work still reaches a rejected
effect. \system separates three questions: what must close, whether the model
reaches the effect frontier, and whether the interface can establish that claim
for every execution the contract allows while required work completes.

Six paths expose missing controls, hidden active work, and models that end too
early. Kafka shows a correct proof can stop before append. Blocking new use of
revoked authority and draining earlier in-flight work closes the studied path
without blocking unrelated work. Machine checking begins with the finite
contract; coverage rests on evidence.

A closure claim must name its subject, kind, policy, frontier, scope, and
configuration and be backed by the controls and tracking needed to make it true.
Within the declared scope, authorization ends for a lineage exactly when it can
issue no new authorization and none already issued can still reach a rejected
effect. 

\bibliographystyle{IEEEtran}
\bibliography{references}

\appendices
\normalsize
\setcounter{dbltopnumber}{4}
\renewcommand{\dbltopfraction}{0.95}
\renewcommand{\dblfloatpagefraction}{0.80}
\section{Open Science and Artifact Availability}
\label{app:artifact}
\urlstyle{tt}
The artifact at \url{https://anonymous.4open.science/r/effectbound/} contains the
implementation, schemas, contracts, evidence records, study protocols, results,
solver, independent checkers, and reconstruction scripts. A manifest records the
submitted snapshot's cryptographic checksum and excludes credentials, private
endpoints, disclosure correspondence, and Git history. After unpacking,
\texttt{make~all} reconstructs the recorded verdicts and runs the formal and
regression checks without contacting live providers. Reruns against managed
providers are optional. The focused target \texttt{make verify-verified-pipeline}
checks all 17 finite games and their certificates. Machine checking starts from
each finite contract. Reproduction cannot establish complete provider-path
coverage or provider conformance. The package omits an
8.1-MiB stress certificate that can be regenerated. Its size, hash, and check
time are recorded.

\section{Ethics Considerations}
\label{app:ethics}

\noindent\textbf{Controlled systems and data.}
The study uses public documentation and source code, controlled fixtures, disposable
Kubernetes, NATS, and Kafka clusters, and synthetic pull requests in a private
research repository. GitHub is the only managed service. We use no production
workloads, nonpublic third-party data, or unsolicited scanning. The GitHub work
comprises inspection of pinned source, 25 controlled REST replay pairs, and one
end-to-end run through the official server. All activity took place in test environments.

\noindent\textbf{Coordinated disclosure.}
We privately reported the GitHub MCP omission of the native reviewed-commit
precondition, NATS purge returning with dispatched callbacks, and Kafka ACL
deletion preceding in-flight appends. We requested review of the interfaces or
documentation. The reports alleged no contract violations and made no
provider-wide claims. At submission, NATS maintainers classified the separation
as intended. GitHub and Kafka discussions remained open. We did not report
documented Kubernetes behavior. Responses refined wording. The analysis
determines verdicts.

\noindent\textbf{Privacy and dual use.}
The sanitized artifact excludes credentials, private endpoints, correspondence,
and identifying metadata. It supports defensive analysis without discovering
paths or providing exploits. The analyst supplies the contract, policy, and scope.

\section{Formal Contract Semantics and Soundness}
\label{app:finite-detail}

This appendix defines information states, provenance, closure, certificates,
and deployment refinement. Grounding the finite contract remains empirical.

\begingroup
\setlength{\abovedisplayskip}{6pt plus 2pt minus 2pt}
\setlength{\belowdisplayskip}{6pt plus 2pt minus 2pt}
\setlength{\abovedisplayshortskip}{3pt plus 2pt}
\setlength{\belowdisplayshortskip}{4pt plus 2pt minus 2pt}

\subsection[Exact Information States and Decision Segments]{Exact Information
States and\\Decision Segments}

The compiler first maps observed history $o$ to the information set
\begin{equation}
\begin{aligned}
\mathsf{Bel}_I(o)=\{(w,s)\mid{}&
w\in W_P,\ \exists s_0\in\mathsf{Init}_P(w),\xi:\\[-1mm]
&s_0\xrightarrow[\,w\,]{\xi}_P s,\
\pi_I(\lambda(\xi))=o\}.
\end{aligned}
\end{equation}
The finite state $q_I(o)$ records this set and all derived controls, successors,
safety, completion, provenance, carrier, and closure data. Equal states must
agree on these data. In particular,
\begin{equation}
q_I(o_1)=q_I(o_2)\Longrightarrow
\Gamma_I(o_1)=\Gamma_I(o_2),
\label{eq:quotient-availability}
\end{equation}
and the same must hold for every compiled predicate. Let $\mathsf{Pref}(q)$ denote
represented prefixes and $\mathsf{Hist}(q)=\{\eta(p)\mid p\in\mathsf{Pref}(q)\}$
their effect histories. Write $\widehat q_I(p)=q_I(\pi_I(\lambda(p)))$ and
$\Gamma(q)$ for the controls induced on $q$.

$\mathsf{Seg}_I(p,u,\sigma,p')$ holds when the mediator issues $u$ at $p$,
$p'=p\cdot u\cdot\sigma$, and hidden moves $\sigma$ end at the next preventive
choice or the declared execution bound. The compiler sets
\begin{equation}
\begin{aligned}
\mathsf{Out}(q,u)
 &=\{\widehat q_I(p')\mid\exists p\in\mathsf{Pref}(q),\sigma:
 \mathsf{Seg}_I(p,u,\sigma,p')\},\\[-1mm]
\mathsf{Cell}_I(q,u)
 &=\{pu\sigma_{\leq t}\mid
 \exists p\in\mathsf{Pref}(q),\sigma,p':\\[-1mm]
 &\hspace{18mm}\mathsf{Seg}_I(p,u,\sigma,p')
 \land 0\leq t\leq|\sigma|\}.
\end{aligned}
\label{eq:compiler-segments}
\end{equation}
Thus $\mathsf{Out}$ gives endpoints and $\mathsf{Cell}_I$ gives all intermediate
execution prefixes. A cell is unsafe if any represented prefix
$p'\in\mathsf{Cell}_I(q,u)$ fails $\mathsf{SafePrefix}$. The executable checker
consumes this Boolean label. It does not reconstruct provider histories.

$\mathsf{Out}(q,u)$ covers modeled success, rejection, failed precondition,
unavailability, and a hang at the declared bound. Exactness preserves each prefix's controls,
outcomes, safety, goals, and the issuance, carrier, suffix, and policy relations
from which closure is derived. Multiple initial views use an
observation-indexed root. Otherwise $q_0$ is unique.

\noindent\textbf{Post-settlement invariant.}
The finite semantic source---the machine-readable contract consumed by Lean---supplies
a response flag, candidate set $\mathsf{Inv}$, and
admitted post-settlement successor relation $\rightarrow_{ps}$. It does not
supply a closure judgment. Every state in the candidate invariant must have an
explicit stutter self-loop. Lean rejects a candidate without these loops.
$\mathsf{Settled}(s)$ records the response. $\mathsf{PolicySafe}(s)$ holds
when every represented prefix satisfies $\mathsf{SafePrefix}$ and the stored
cell label records no unsafe intermediate prefix. $\mathsf{ClaimSafe}(s)$ checks
every claim emitted by that point.

Their information-state forms hold when the corresponding
predicate holds for every represented concrete state. Let
$\mathsf{Entry}(s)$ mean
that $s$ has emitted its declared response and settled safely, every observed
claim is true, and, if its world lies in $W_A$, its workload and required claims
have completed. A world in $W_P\setminus W_A$ may instead settle through a
declared safe rejection. Lean accepts the candidate only if
\begin{equation*}
\begin{aligned}
\mathsf{Entry}(s)&\Rightarrow s\in\mathsf{Inv},\\[-1mm]
s\in\mathsf{Inv}&\Rightarrow
  \mathsf{Settled}(s),\\[-1mm]
s\in\mathsf{Inv}&\Rightarrow
  \mathsf{PolicySafe}(s)\land\mathsf{ClaimSafe}(s),\\[-1mm]
s\in\mathsf{Inv}&\Rightarrow s\rightarrow_{ps}s,\\[-1mm]
\rightarrow_{ps}&\subseteq\mathsf{Inv}\times\mathsf{Inv},\\[-1mm]
s\rightarrow_{ps}s'\land{}&\mathsf{SuccessObserved}(\chi,s)\\[-1mm]
&\Rightarrow\mathsf{SuccessObserved}(\chi,s').
\end{aligned}
\end{equation*}
The final conditions confine $\rightarrow_{ps}$ to
$\mathsf{Inv}\times\mathsf{Inv}$, provide an explicit stutter transition at
every invariant state, and preserve every emitted claim.
Thus $\mathsf{ContractSettled}(q)$ requires every represented state to emit its
response and enter this invariant. Progress ends at entry. The invariant
preserves the claim. The successor inventory remains empirical. Lean checks
membership, inductivity, policy safety, and claim preservation. Neither
settlement nor invariant membership implies carrier, provider, or system
quiescence.

\subsection{Configuration History, Provenance, and Live Carriers}

Let $\mathcal H_\beta=((\beta_0,M_0),\ldots,(\beta_k,M_k))$ be the installed
history, with current pair $\operatorname{last}(\mathcal H_\beta)=(\beta,M)$.
Here $\beta$ identifies a configuration version and $M$ its strategy.
$\mathsf{Pref}_{\mathcal H_\beta}(q)$ contains admitted prefixes consistent with
the strategy active at each point. Appending $(\beta',M')$ does not re-filter
the past. We omit the history when clear.
For a claim emitted under version $\beta_e$, $q^{\beta_e}$ denotes the same
information state and full later configuration history, with $\beta_e$ marked
as the claim's configuration version. Closure predicates are state-indexed and do not
carry a separate configuration superscript. We write $q$ when the marked version
is clear.

For every issuable $a$, the analyst supplies a nonempty finite family
$\mathsf{Prov}(a)$ of nonempty lineage sets. Membership within one set is
conjunctive, while alternative sets are disjunctive.
The abstract semantics permits this relation. The evaluated contracts use the functional case
$\mathsf{Prov}(a)=\{\{g\}\}$.
The workload and policy also fix the normative relation
$\mathsf{Permitted}(a,e)$ introduced in Section~\ref{sec:problem}, with
$\mathsf{AuthScope}(a)=\{e\mid\mathsf{Permitted}(a,e)\}$. Let
$\mathcal K_{\rm mat}$ be the carriers represented in the semantic source. For each
$c\in\mathcal K_{\rm mat}$, the source supplies its authorization instance
$\mathsf{Auth}(c)$, a nonempty set $\mathsf{AuthEffects}(c)$ of effects that
justified its creation or continuation, and nonempty evidence identifiers
$\mathsf{EvidenceRefs}(c)$ for that handoff.
Lean checks the first three conditions for every represented carrier, live or
not, and the final two for every represented $\mathsf{Carries}$ fact:
\begin{equation*}
\begin{gathered}
\mathsf{AuthEffects}(c)\ne\varnothing,\\[-1mm]
\mathsf{AuthEffects}(c)\subseteq
 \mathsf{AuthScope}(\mathsf{Auth}(c)),\\[-1mm]
\mathsf{EvidenceRefs}(c)\ne\varnothing,\\[-1mm]
\mathsf{Carries}(c,a,H,q)\Rightarrow c\in\mathcal K_{\rm mat},\\[-1mm]
\mathsf{Carries}(c,a,H,q)\Rightarrow a=\mathsf{Auth}(c).
\end{gathered}
\end{equation*}
These checks establish link well-formedness. Evidence and physical outcomes remain grounded.

Each represented suffix $z$ records its carrier $c_z$ and physical effect $e_z$
separately. Let $e(\tau,t)$ be the effect added when suffix $\tau$ crosses $F$ at step $t$.
Write $\mathsf{SuffixCrosses}_{F}(q,z)$ when that suffix crosses $F$. Then
$\mathsf{MayEffect}(c,q)=\{e_z\mid c_z=c\land
\mathsf{SuffixCrosses}_{F}(q,z)\}$. In particular,
$\mathsf{MayEffect}(c,q)$ need not be a subset of
$\mathsf{AuthScope}(\mathsf{Auth}(c))$. GitHub's moved-head world has
$\mathsf{Carries}(c,a,\ldots)$ while the reachable merge of $h_2$ lies outside
the permission for $h_1$. Within $\kappa$, each
$\mathsf{SuffixCrosses}_{F}(q,z)$ corresponds to an admitted finite witness
$(c,a,H,\tau,t)$: $c_z=c$, $e_z=e(\tau,t)$, and its recorded verdict equals
$\varphi_{\rm app}(H\cdot\eta(\tau_{<t})\cdot e_z)$. Grounding supplies this
correspondence both ways between represented suffixes and declared in-scope
witnesses. Lean checks finite consistency. Provider completeness remains empirical.
Missing permission semantics, handoff evidence, or suffix semantics yields
\textsc{Unsupported}. An explicit denial makes the effect unsafe. Mere causal
position is insufficient.
To include every declared possible authorization origin, define
\begin{equation}
\mathsf{DerivedFrom}(a,g)\iff
\exists J\in\mathsf{Prov}(a):g\in J.
\end{equation}
Thus $(g_1\land g_2)\lor g_3$ makes $g_1,g_2$ joint origins and $g_3$ an
alternative. Within the declared inventory, $\mathsf{IssuedLedger}(q)$ records
exactly the represented instances issued at compatible prefixes and only grows: a
change from $\beta$ to $\beta'$ copies the ledger before new issues and never
removes an entry:
\begin{equation}
\begin{aligned}
\mathsf{IssuedFrom}(a,g,q)\iff{}&
\mathsf{DerivedFrom}(a,g)\\[-1mm]
&{}\land a\in\mathsf{IssuedLedger}(q).
\end{aligned}
\label{eq:issued-from}
\end{equation}

The live-carrier store $\mathsf{Live}(q)$ contains tuples $(c,a,H,p_c)$, with
$\mathsf{At}_I(c,a,H,q,p_c)\iff(c,a,H,p_c)\in\mathsf{Live}(q)$.
$\mathsf{Pending}_I(c,a,q,p_c)$ means that a supported creation of $c$ carrying
$a$ is represented at current prefix $p_c$, with no earlier event ending it.
Exactness is
\begin{equation}
\begin{aligned}
&(c,a,H,p_c)\in\mathsf{Live}(q)\\[-1mm]
&\quad\iff \mathsf{Pending}_I(c,a,q,p_c)\\[-1mm]
&\qquad{}\land a\in\mathsf{IssuedLedger}(q)
\land H\in\mathsf{Hist}(q)\\[-1mm]
&\qquad{}\land\mathsf{Carries}(c,a,H,q)\\[-1mm]
&\qquad{}\land p_c\in\mathsf{Pref}_{\mathcal H_\beta}(q)
\land H=\eta(p_c).
\end{aligned}
\label{eq:live-store-exact}
\end{equation}
A configuration-only switch extends
$\mathcal H_{\beta'}=\mathcal H_\beta\mathbin{\cdot}(\beta',M')$ and, before any
provider move, copies both stores:
$\mathsf{IssuedLedger}(q^+)=\mathsf{IssuedLedger}(q^-)$ and
$\mathsf{Live}(q^+)=\mathsf{Live}(q^-)$. Equation~\ref{eq:live-store-exact} at
$q^+$ uses $\mathcal H_{\beta'}$. Creating a carrier inserts its tuple. A
handoff or progress step updates its prefix and history. Only a represented,
evidence-supported end event removes it.

An unlisted creation, ambiguous end event, missing correlation, or unjustified
direction in Equation~\ref{eq:live-store-exact} yields \textsc{Unsupported}.
Configuration changes preserve issued instances and live carriers. Compatible
histories that disagree about $a$ remain separate witnesses. Grouped ACL sets or
epochs form one $g$ only with evidence. Provenance, issuance, live stores,
attempt matching, scope, configuration version, templates, and latched claims must agree for all
histories merged into the same $q$. The solver infers none of these relations.
Lean derives closure from the finite relations. The checker and solver do not
establish that the represented provider inventories are complete. Completeness
remains an empirical obligation.

\subsection{Closure Predicates and Typed Claims}

Let $\mathsf{Fut}_{\mathcal C}(p,g)$ denote the admitted in-scope extensions of
$p$ containing an attempt that matches $g$:
\begin{equation}
\begin{aligned}
&\mathsf{FutureUseClosed}(g,q)\\[-1mm]
&\quad\iff\nexists p,\zeta,a:\quad p\in\mathsf{Pref}(q)\\[-1mm]
&\hspace{20mm}\land\zeta\in\mathsf{Fut}_{\mathcal C}(p,g)\\[-1mm]
&\hspace{20mm}\land\mathsf{Issues}(\zeta,a)\\[-1mm]
&\hspace{20mm}\land\mathsf{DerivedFrom}(a,g).
\end{aligned}
\label{eq:future-use-closed}
\end{equation}

For already-issued authority, $\mathsf{Tail}_I$ records provider moves until the
next mediator choice:
\begin{equation}
\begin{aligned}
&\mathsf{Tail}_I(c,a,H,q)\\[-1mm]
&=\{\tau\mid\exists q_b,p,u,\sigma,p',k,p_c:\\[-1mm]
&\quad q_b=\widehat q_I(p)
 \land p\in\mathsf{Pref}_{\mathcal H_\beta}(q_b)\\[-1mm]
&\quad{}\land u=M_{\mathcal H_\beta}(p)\\[-1mm]
&\quad{}\land\mathsf{Seg}_I(p,u,\sigma,p')\\[-1mm]
&\quad{}\land 0\leq k\leq|\sigma|\\[-1mm]
&\quad{}\land p_c=pu\sigma_{<k}
 \land\tau=\sigma_{\geq k}\\[-1mm]
&\quad{}\land\mathsf{At}_I(c,a,H,q,p_c)\}.
\end{aligned}
\label{eq:compiler-tails}
\end{equation}
Tail $\tau$ has no mediator action before the next preventive choice. At a
frontier crossing, $e(\tau,t)\in\mathcal E$ is the effect record defined above.
Define
\begin{equation*}
\begin{aligned}
\mathsf{AuthorizedSafe}(a,H',e)
\equiv{}&\mathsf{Permitted}(a,e)\\[-1mm]
&{}\land(\varphi_{\rm app}(H'\cdot e)=1).
\end{aligned}
\end{equation*}
Carrier $c$ is open iff such a tail crosses the frontier and fails this
authorization-and-policy test:
\begin{equation}
\begin{aligned}
&\mathsf{Open}_{\varphi_{\rm app}}(c,a,H,q)\iff
\mathsf{Carries}(c,a,H,q)\\[-1mm]
&\quad{}\land\exists\tau\in\mathsf{Tail}_I(c,a,H,q),\,t:\\[-1mm]
&\quad{}\mathsf{Crosses}_{F}(\tau,t)
\land\neg\mathsf{AuthorizedSafe}
  (a,H\cdot\eta(\tau_{<t}),e(\tau,t)).
\end{aligned}
\label{eq:open-continuation}
\end{equation}
An instance is effect-closed iff no represented history compatible with $q$ has
an open carrier:
\begin{equation}
\begin{aligned}
&\mathsf{EffectClosed}_{\varphi_{\rm app}}(a,q)\\[-1mm]
&\quad\iff\nexists H,c:\quad H\in\mathsf{Hist}(q)\\[-1mm]
&\hspace{29mm}\land\mathsf{Open}_{\varphi_{\rm app}}(c,a,H,q).
\end{aligned}
\label{eq:effect-closed}
\end{equation}
Projection is compositional: $\eta(xy)=\eta(x)\eta(y)$. Any tail witness has
$H=\eta(p_c)$ and $H\eta(\tau_{\leq t})=\eta(p_c\tau_{\leq t})$ before the next
observation. If $u$ crosses $F$, $\mathsf{Cell}_I$ checks its $t=0$ prefix for
policy safety. $\mathsf{CanCross}_F(c,a,H,q)$ abbreviates the existence of
$\tau\in\mathsf{Tail}_I(c,a,H,q)$ and $t$ such that $\mathsf{Crosses}_F(\tau,t)$. Stronger
quiescence is
\begin{equation}
\begin{aligned}
&\mathsf{Quiescent}(a,q)\\[-1mm]
&\quad\iff\nexists H,c:\quad H\in\mathsf{Hist}(q)\\[-1mm]
&\hspace{19mm}\land\mathsf{Carries}(c,a,H,q)\\[-1mm]
&\hspace{19mm}\land\mathsf{CanCross}_{F}(c,a,H,q).
\end{aligned}
\label{eq:quiescent}
\end{equation}

Write $\bar\chi@(\kappa,\beta)$ for a typed template instantiated under scope
$\kappa$ and configuration version $\beta$. Exact matching requires
$\mathsf{Match}_{\kappa,\beta}(\bar\chi,\chi)$ iff
$\chi=\bar\chi@(\kappa,\beta)$.
For claim scope $\kappa'$, interpretations are
\begin{equation*}
\begin{aligned}
&\mathcal I_{\varphi_{\rm app},\kappa}(\mathsf{FutureUseClosed}(g)@
  (\kappa',\beta),q)\\[-1mm]
&\quad\equiv(\kappa'=\kappa)\land
  \mathsf{FutureUseClosed}(g,q^{\beta}),\\[-1mm]
&\mathcal I_{\varphi_{\rm app},\kappa}(\mathsf{InstanceEffectClosed}(a)@
  (\kappa',\beta),q)\\[-1mm]
&\quad\equiv(\kappa'=\kappa)\land
  \mathsf{EffectClosed}_{\varphi_{\rm app}}(a,q^{\beta}),\\[-1mm]
&\mathcal I_{\varphi_{\rm app},\kappa}(\mathsf{LineageClosed}(g)@
  (\kappa',\beta),q)\\[-1mm]
&\quad\equiv(\kappa'=\kappa)\land\mathsf{FutureUseClosed}(g,q^{\beta})\land{}\\[-1mm]
&\qquad\forall a:\ \mathsf{IssuedFrom}(a,g,q^{\beta})\Rightarrow\\[-1mm]
&\hspace{29mm}\mathsf{EffectClosed}_{\varphi_{\rm app}}(a,q^{\beta}).
\end{aligned}
\end{equation*}
After a switch to $\beta'$, $q^{\beta}$ retains the complete later configuration
history. Its quantifier therefore covers instances and carriers inherited from
$\beta$ as well as those issued after the switch. A claim emitted after the
switch uses $\beta'$. Claims emitted under $\beta$ remain tied to $\beta$ and the
original scope $\kappa$, so $\kappa'\ne\kappa$ fails the exact match. Each
evaluated contract covers one configuration version. Preserving a claim emitted
under an earlier version after a later one is installed requires a separate
deployment proof. The finite checker records the later history but does not
infer cross-version preservation from it.

\subsection{Finite Decision, Certificates, and Obstruction}

If $q_0\in\mathcal W$, each ranked non-goal state has an action whose successors
have lower rank. Policy and claim safety hold until $G$. In validated games,
every available action has an admitted response. Outside $\mathcal W$, a state
is inadmissible or every available action has a successor outside $\mathcal W$.
Repeating this argument rules out every observation-based strategy.

$\mathsf{Twin}_I(w_A,w_U)$ holds when the provider can keep the worlds
observationally equivalent against every mediator restricted to $I$ until the
required completion becomes impossible in $w_A$ or a forbidden prefix occurs in
$w_U$. Applying this condition to the universal predecessor gives
Section~\ref{sec:formal}'s authorization-twin obstruction.

Sure-winning randomization adds no power: under
$\mathsf{SureWin}_{\rm rand}(q_0)$, every action in the strategy's support
satisfies the universal predecessor, so choosing one wins. A deterministic
strategy is the special case that chooses one action with probability~1. Thus
\begin{equation}
\mathsf{SureWin}_{\rm rand}(q_0)
\quad\Longleftrightarrow\quad
q_0\in\mathcal W.
\label{eq:purification}
\end{equation}
We make no expected utility or almost-sure claim.

\noindent\textbf{Verified translation and executable certificates.}
Machine checking begins with a finite semantic source built from the evidence.
An untrusted compiler proposes the finite game, state mapping, and reachability
ranks. The source supplies the Boolean value of
$\neg\mathsf{SafePrefix}$ for each represented base prefix. Machine checking
uses this value without reconstructing provider histories or reevaluating
$\varphi_{\rm app}$. For each represented frontier-crossing suffix, Lean combines
$\mathsf{Permitted}$ with the recorded policy verdict.
Lean checks source well-formedness, the post-settlement invariant, exact initial
and reachable states, observations, safety including $\mathsf{ClaimSafe}$,
required goals and claims, actions, responses, and successor coverage, then
constructs the finite information-state game.

The executable checker consumes the serialized certificate for the proposed
game. For a winning certificate, it checks the ranked region, actions, nonempty
responses, and decreasing ranks. For a losing certificate, $L$ is the certified
set. The checker verifies $q_0\in L$ and $q\notin G$ for every $q\in L$. At each
admissible $q\in L$, every available action must have an admitted successor in
$L$. Non-admissible states require no witness.

Acceptance establishes the winning condition for a winning certificate and
rules it out for a losing certificate. After the semantic source is constructed,
neither compiler nor solver is trusted. Lean proves the executable checker sound. A
separate Python check verifies the certificate inventory, origin hashes, and totals
of 2 winning and 15 losing verdicts.

Lean checks the closure definitions and synthetic examples that separate the
predicates and exercise provenance, including $(g_1\land g_2)\lor g_3$. Its
executable accepts all 17
published finite-game certificates.
The four evaluated NATS contracts are the purge base and its Bind,
native-terminal, and Tracked Drain variants. Each declares the same
$\mathsf{InstanceEffectClosed}(a_{\rm purged\_delivery})$ requirement. Lean
checks the materialized scope and configuration match flags. From the supplied
issuance relation, carrier provenance, frontier-crossing suffixes, permission
relation, application-policy verdicts, and coverage relations, it derives whether
the requirement and resulting $\mathsf{ClaimSafe}$ condition hold. Bind and Tracked Drain satisfy the
requirement. The base and native-terminal variants fail it. The checker supports
future-use, instance-effect, and lineage claims, while these evaluated variants
exercise instance-effect closure. In these contracts, the scoped coverage
assertion states that every live carrier and every admitted frontier-crossing
suffix of those carriers is represented in the declared scope. The 14
missing-effect-history records are grounding diagnostics, inventoried separately
from game verdicts. The checker takes these coverage assertions and source
relations as semantic inputs. Reconstructing them from provider traces remains an empirical obligation.
The proof does not establish provider grounding, completeness of the
supplied route and successor inventories, or artifact provenance.

\subsection{Deployment-to-Contract Soundness}

The deployed mediator follows the repaired strategy and implements the boundary
change. Refinement preserves observations, effects, responses, claims, and
completion events, so contract predicates apply to each refined trace.

\begin{corollary}[Deployment-to-contract soundness]
Let $\mathcal C$ be a finite provider-boundary contract, $\Delta$ a candidate
repair, and $M_\Delta$ the winning strategy for $\mathcal C\oplus\Delta$ under
configuration version $\beta_\Delta$. Assume
$\mathsf{PreservesExternalContract}(\mathcal C,\Delta)$, that the deployed mediator follows
$M_\Delta$, and that every admitted trace of the deployed provider boundary
refines $\mathcal C\oplus\Delta$.
Then every prefix is safe, and each visible
closure response satisfies its claim under $(\kappa,\beta_\Delta)$. Each corresponding required world under
$\iota_\Delta$ observes a matching claim for every $\Xi$ template by required
completion. Equivalently, the refined executions satisfy
$\mathsf{PolicySafe}_{\iota_\Delta(W_P)}(M_\Delta)$,
$\mathsf{ClaimsSafe}_{\iota_\Delta(W_P)}(M_\Delta)$,
$\contractsettled_{\iota_\Delta(W_P)}(M_\Delta)$, and
$\requiredcomplete_{\iota_\Delta(W_A)}(M_\Delta)$.
\end{corollary}

\emph{Proof sketch.}
Refinement confines execution to the repaired contract; the deployed mediator applies its
winning strategy, and admissibility enforces policy safety and
$\mathsf{ClaimSafe}$. Apply Theorem~\ref{thm:realizability}.

A checked trace establishes its own prefix safety and emitted closure claims;
$\contractsettled_{\iota_\Delta(W_P)}(M_\Delta)$ and
$\requiredcomplete_{\iota_\Delta(W_A)}(M_\Delta)$ across worlds require
a refinement proof over all executions or mediation of every in-scope execution.

\endgroup

\section{Evidence, Scope, and Validation}
\label{app:confidence-detail}

The artifact retains the per-run records, manifests, preregistration material,
selection hashes, and verifier outputs summarized here.

\noindent\textbf{Scope and performance.}
The 644-line Kafka prototype instruments \texttt{handleProduceAppend}. Kafka
remains unmodified. Its three-broker fixture processed 1,310,720 records in 160
error-free trials. At $c=(32,128)$, Stock and full fixed-set throughput medians
are $(10.328,10.319)$ and $(8.340,8.281)$ thousand records/s, with p99 medians
of $(82.0,279.0)$ and $(170.7,302.2)$ ms. Full fixed-set throughput ranges are
$(-40.1$--$-1.9,-42.4$--$+23.4)\%$ and p99 ranges are
$(+36.5$--$+156.9,-105.3$--$+133.7)$ ms. The direct Stock--old-instance-fence
contrast was computed after data collection from the same matched observations.
The artifact reports intermediate variants, startup and host conditions,
old-instance-fence ranges, and separate recovery timings.

\noindent\textbf{Coverage and abstention.}
The 12 contracts comprise six \textsc{Unrealizable} path contracts with repairs,
two \textsc{Realizable} controls, and four \textsc{Unsupported} cases.
The abstention audit ranks each stratum's identifiers by SHA-256 over seed,
stratum, and identifier with NUL separators. Because the author knew the seed
when framing candidates, selection is deterministic but neither independent nor
representative. Seven of eight yield \textsc{Unsupported}, without estimating
applicability or prevalence. In a preregistered 20-operation audit, SDKs retain
all native primitives, while tools preserve 11, omit two, misplace five, and weaken
two. In the NATS time-bound negative control, a retry after 0.8\,s exceeds the
0.5\,s duplicate window and commits twice. Kafka traces link request identity,
authorization epoch, broker boot, deletion, return, and local append but cover
only observed executions.

\noindent\textbf{Falsification and finite checks.}
The evidence graph has 95 obligations and 41 records. Each of the 38 provider
premises used in grounded verdicts has two support types. Of 33 eligible
leave-one-out removals, 25 yield \textsc{Unsupported}. All 30 critical mutations
force abstention, while six benign edits preserve the result. Frontier,
time-bound, repair-set, and exclusion mutations alter results as expected.

\end{document}